%% file: Ready_Submit_JPV.tex
\documentclass[final]{econsocart}

\input{prelim}
\RequirePackage[colorlinks,citecolor=cyan,linkcolor=blue,urlcolor=blue,pagebackref]{hyperref}
\usepackage{booktabs,tabularx}
\begin{document}
\begin{frontmatter}
	
	\title{Justifiable Priority Violations}
	\runtitle{Justifiable Priority Violations}
	
	\begin{aug}
		\singlespacing 
		\author[id=au1,addressref={add1}]{\fnms{Josu\'e}~\snm{Ortega}}
		\author[id=au3,addressref={add3}]{\fnms{R. Pablo}~\snm{Arribillaga}}
		
		%%%%%%%%%%%%%%%%%%%%%%%%%%%%%%%%%%%%%%%%%%%%%%
		%% Addresses                                %%
		%%%%%%%%%%%%%%%%%%%%%%%%%%%%%%%%%%%%%%%%%%%%%%
		\address[id=add1]{\orgname{Queen's University Belfast}}
		\address[id=add3]{\orgdiv{Instituto de Matemática Aplicada San Luis},
			\orgname{Universidad Nacional de San Luis, and CONICET}}
		
	\end{aug}
	
	\support{Emails: j.ortega@qub.ac.uk, rarribi@unsl.edu.ar.
		Date: {\it {\ddmmyyyydate\today}}.\\
		We are grateful to four anonymous reviewers for their thoughtful comments.
		We also thank three anonymous reviewers from 22nd Conference on Web and Internet Economics, where this paper was accepted.
		Bettina Klaus, Cyril Rouault and Juan Sebastián Pereyra gave us exceptional feedback, as well as audiences at the 37th Stony Brook International Conference on Game Theory and the Economics Seminar at Universidad de Montevideo. }
	
	%\coeditor{\fnm{[Name} \snm{Surname}; will be inserted later]}

	\begin{abstract}
		\begin{center}
			\bf Abstract
		\end{center}
		Making Deferred Acceptance more efficient requires priority violations, but which ones are justifiable?
		We argue that a priority can be justifiably violated if the affected student either (i) directly benefits from the improvement that caused said violation, or (ii) is unimprovable under any assignment that Pareto-dominates DA.
		We construct a ``\emph{just-below-cutoffs}'' mechanism that always finds a justifiable DA-improvement whenever DA's outcome is inefficient, and build on it to construct an algorithm that expands justifiable improvements iteratively, converging to a matching that cannot be Pareto-improved by any justifiable outcome.
		Finally, we prove that both justifiability and consent frameworks have important limitations in reaching Pareto-efficient outcomes, and use data and simulations to quantify how binding these constraints are in practice.
	\end{abstract}

	
\end{frontmatter}

\setlength{\parskip}{0.25em} 
%%%%%%%%%%%%%%%%%%%%%%%%%%%%%%%%%%%% 
\section{Two Paths Beyond Deferred Acceptance} 
\label{sec:introduction} 
%%%%%%%%%%%%%%%%%%%%%%%%%%%%%%%%%%%%%%%%%
A school choice mechanism must balance two often conflicting objectives: respecting legally defined priorities and allocating scarce school places efficiently.
The student-proposing Deferred Acceptance algorithm (DA) has become the dominant benchmark because it produces the student-optimal stable matching: an allocation that respects priorities that is not Pareto-dominated by any other stable matching.
Nonetheless, DA can generate sizeable welfare losses, so it is natural for an Education Authority to consider alternative mechanisms that Pareto-improve upon the DA outcome.\footnote{The Education Authority of Flanders is one such example; see \citet{cerrone2022school}.}
Any such improvement necessarily entails priority violations, but which violations are justifiable?
Our contribution is to propose a new concept of controlled priority violations called \emph{justifiable} that allows us to improve upon DA whenever its outcome is Pareto-inefficient.

To do so, we divide students with violated priorities into three categories: (1) beneficiaries, who improve their DA placement; (2) unimprovable students, whom no DA-improvement could help; and (3) improvable non-beneficiaries, who could have benefited under some DA improvement but who do not under the chosen one.
In our view, the first two have no grounds for complaint: they either gained or could never have gained.
The third category is different, as such a student could object: \emph{why was my priority overridden if I did not gain, and I could have?}

This potential objection motivates our justifiability criterion: a DA improvement is \emph{justifiable} if it violates priorities only for beneficiaries and unimprovable students, but never for improvable non-beneficiaries.
The criterion is endogenous: whether a violation is justifiable depends on what the improvement accomplishes for the affected student relative to what was feasible, not on permissions granted in advance.

We show that justifiable improvements always exist when DA is inefficient (Theorem \ref{thm:existence}) and characterize a benchmark obtained by the \emph{Just-Below-Cutoffs} (JBC) mechanism.
JBC identifies, for each school that rejected an improvable student, the highest-priority improvable student below its DA cutoff.
Each school points to that student's DA school, and JBC executes the resulting cycles, so no improvable student's priority is violated by its trades.
We call this type of justifiable improvement \emph{strongly justifiable}.
Any matching satisfying this stronger notion corresponds to executing disjoint improvement cycles.
These strongly justifiable matchings are essentially unique: if a student improves in two strongly justifiable outcomes, she gets assigned to the same school in both.
JBC selects the largest collection of such cycles.

Despite its normative appeal, justifiability is generally incompatible with full Pareto-efficiency (Proposition \ref{prop:efficiency}): achieving full efficiency may require violating improvable non-beneficiaries' priorities, precisely the violations our criterion excludes.
\josue{To achieve constrained efficiency while preserving justifiability, we construct a polynomial-time algorithm, SJBC+, that builds on JBC via a sequential expansion, and which produces a justifiable matching that is not Pareto dominated by any other justifiable outcome.}

\josue{At a high level, SJBC+ starts from JBC and its beneficiary set.
	It then enlarges the set of admissible trades by allowing priority violations not only against unimprovable students, but also against students who have already benefited.
	At each iteration, it preserves all existing beneficiaries while finding a matching that maximizes the number of students receiving a better-than-DA assignment.
	Newly benefited students then enlarge the set of priorities that may be justifiably relaxed in the next iteration.
	Once this beneficiary set stops expanding, the ``$+$'' step searches for further Pareto improvements among the final beneficiaries.
	The resulting outcome is justifiable and cannot be Pareto-dominated by another justifiable matching (Theorem~\ref{thm:stpcplus}).
	Each expansion step is a standard assignment problem on an augmented bipartite graph, so SJBC+ is polynomial-time computable.}

Justifiability differs fundamentally from the dominant alternative for controlled priority violations: school choice with ex‑ante consent.
In that approach, the mechanism designer asks students in advance if their priority could be overridden in case doing so benefits others at no cost for themselves.
The set of consenting students is fed into a consent-based mechanism, such as Efficiency-Adjusted DA (EADA), which always returns a DA improvement that respects the priorities of non-consenting students.
We show that justifiability allows for improvements unattainable by EADA under any consent set, just as EADA's outcomes need not be justifiable (Theorem \ref{prop:cc_vs_eada};  our simulations in Section \ref{sec:simulations} show that this occurs very frequently).
While a number of other weaker stability notions also allow for controlled priority violations, such as legality and priority-neutrality, they converge to EADA with full consent as the student-optimal matching within these restrictions; thus, justifiability departs from these normative frameworks too.

More importantly, despite EADA's full efficiency guarantee when all students consent, we show that EADA and any other consent-based mechanism also faces important limitations when only a fraction of students consent.
In particular, any consent-based mechanism that incentivizes consent (i.e. does not give a worse placement to any student who consents) and is constrained efficient (i.e. its outcome is not Pareto-dominated by another matching which respects the priorities of non-consenting students), fails to be \emph{efficient-whenever-possible}: there are Pareto-efficient matchings that respect the priorities of non-consenters, yet the mechanism (designed to address DA's Pareto-inefficiency) returns itself a Pareto-inefficient allocation (Theorem \ref{thm:bigimpo}).
Taken together, our results show the complexity of guaranteeing full efficiency when improving upon DA on either the justifiability or the consent-based frameworks.

To understand how frequently our theoretical results bite, we use simulations and administrative data from 64,893 Chilean students to compare SJBC+ with EADA under full and empirically observed consent \josue{(under the assumption of truthful preference revelation)}.
The results reveal a clear welfare/priority-protection tradeoff.
SJBC+ generates substantial welfare gains while restricting all priority violations to the justifiable class.
EADA under partial consent protects priorities more strongly but improves fewer students, while full-consent EADA provides a useful welfare benchmark at the cost of priority violations that need not be justifiable.

In summary, our paper proposes a new normative notion that allows for defensible priority violations for Education Authorities that would like to obtain a matching that respects most priorities but that allows for controlled priority violations in order to reach more efficient outcomes.
Our results show the limits of this framework with regard to how much efficiency it can recover in theory and in practice.
Importantly, the concept of justifiability allows for improvements that cannot be obtained by existing approaches, and thus provides a new direction for the school choice literature.

\paragraph{Outline.}
Section \ref{sec:literature} discusses the literature.
Section \ref{sec:model} introduces the model.
Section \ref{sec:results} develops the theoretical results on justifiable DA improvements.
Section \ref{sec:ewp} discusses the relationship between justifiability and school choice with consent.
Section \ref{sec:simulations} analyzes the performance of our proposed algorithm using simulations and real-life school choice data.
Section \ref{sec:incentives} discusses incentive properties.
Section \ref{sec:conclusion} concludes.

%%%%%%%%%%%%%%%%%%%%%%%%%%%%%%%%%%
\section{Literature}
\label{sec:literature}
%%%%%%%%%%%%%%%%%%%%%%%%%%%%%%%%%%%%%%%%%%

The student-proposing Deferred Acceptance mechanism is often Pareto-inefficient for students \citep{abdulkadirouglu2003}, yet any Pareto improvement over it must necessarily violate some student's priority.
This tension has motivated a large literature on which priority violations are acceptable when improving upon DA.

A striking feature of this literature is that several conceptually different normative frameworks end up coinciding on the same foundational adjustment, Efficiency-Adjusted Deferred Acceptance \citep[EADA;][]{kesten2010school}.
Kesten's original justification is ex-ante consent: a priority violation is acceptable when the affected student agrees in advance to waive her priority.
\citet{reny2022efficient} proposes priority neutrality, requiring (roughly) that a displaced student can only lose priority when the violation is in a precise sense unavoidable for helping others; in that framework, EADA delivers the unique student-optimal priority-efficient outcome.
\citet{ehlers2020legal} propose legality, a set-valued stability concept in which blocking is admissible only if the student can be assigned to the blocking school in some assignment under consideration.
They show there is a unique student-optimal legal assignment which coincides with EADA's outcome.\footnote{EADA's properties have been extensively studied; see \cite{bando2014existence, tang2014new, dur2019school, troyan2020essentially, troyan2020obvious, tang2021weak, chen2023regret, shirakawa2024simple}.}
Our notion of justifiability departs from this literature: it yields justifiable matchings that lie outside the range of EADA under any consent structure, and thus need not be priority-neutral nor legal.

Our mechanisms build on the improvement-cycle tradition, but target a different constraint.
The Just-Below Cutoffs (JBC) algorithm is related to the Top Priority algorithm \citep{dur2019school}, which iteratively resolves improvement cycles while allowing violations against any student who consents ex-ante; in contrast, JBC uses a single-step construction that violates only {unimprovable} students' priorities.
JBC is also closely related to stable improvement cycles \citep{erdil2008s}, which address inefficiency driven by arbitrary tie-breaking under weak priorities: starting from a stable outcome under a fixed tie-break, they implement Pareto-improving cycle exchanges that remain stable with respect to the underlying weak priorities.
JBC targets a different source of inefficiency: it improves upon DA even under strict priorities, but only by violating the priorities of unimprovable students.

Several other strands are complementary to ours.
\citet{kitahara2021improving} analyze how to improve upon DA under partial priority structures.
Other work studies when DA can be improved while preserving strategy-proofness \citep{kesten2019strategy}, and develops refinements based on consistency requirements \citep{dougan2020consistent}.
Recent work also shows that DA produces an inefficient matching with high probability, and that almost all students are improvable \citep{ortega2025identifying}, and documents limitations that no Pareto improvement over DA can overcome \citep{ortega2025pareto}.
These theoretical results, together with the rural-hospital property satisfied by every matching that Pareto-dominates DA \citep{alva2019stable}, help to explain EADA's empirical performance in counterfactuals using real-life data \citep{ortega2023cost}.
A complementary literature examines how to improve upon DA while minimizing the number of blocking pairs or triplets generated \citep{dougan2021minimally, kwon2020justified, afacan2022improving}.
Finally, our work is also related to \citet{decerf2024incontestable}, who study incontestable assignments that can be defended against appeals; however, in their framework every improvement upon DA is incontestable.

%%%%%%%%%%%%%%%%%%%%%%%%%%%%%%%%%%
\section{Model}
\label{sec:model}
%%%%%%%%%%%%%%%%%%%%%%%%%%%%%%%%%%%%

A many-to-one school choice problem $P$ consists of two finite sets of students $I$ and schools $S$.
Each student $i$ has a strict preference $\succ_i$ over $S\cup\{s_\emptyset\}$, where $s_\emptyset$ denotes being unassigned.
Each school $s$ has quota $q_s\in\mathbb{N}$ and a strict priority order $\triangleright_s$ over students (except for $s_\emptyset$, with $q_{s_\emptyset}=\infty$).\footnote{While our examples use unit capacities for clarity, all of our results hold for arbitrary quotas.}

A matching $\mu$ assigns each student to one school, $\mu:I\rightarrow S\cup\{s_\emptyset\}$, such that $|\mu^{-1}(s)|\le q_s$ for all $s\in S$.
A matching $\mu$ is non-wasteful if, whenever a school $s\in S$ has an empty seat under $\mu$, no student strictly prefers $s$ to her assignment.
Throughout the paper, we restrict attention to non-wasteful matchings.

Let $\rk_i(s)\in\{1,\dots,|S|+1\}$ be the rank of school $s$ under $\succ_i$ (lower is better), and let $\rk_s(i)$ be the priority rank of student $i$ at school $s$ (lower is better priority).
The matching $\mu$ weakly Pareto-dominates $\nu$ if $\rk_i(\mu_i)\le \rk_i(\nu_i)$ for all $i$; it Pareto-dominates $\nu$ if it weakly Pareto-dominates $\nu$ and the inequality is strict for some $i$.
A matching is Pareto-efficient if it is not Pareto-dominated by any other matching.

We say student $j$ violates student $i$'s priority under $\mu$ if there exists a school $s\in S$ such that (i) $\mu_j=s$, (ii) $i$ prefers $s$ to $\mu_i$, and (iii) $\rk_s(i)<\rk_s(j)$.
A matching is stable if it does not violate the priority of any student.

\paragraph{Deferred Acceptance and its Envy Digraph.}
A mechanism maps a problem $P$ to a matching.
The student-proposing Deferred Acceptance (DA) is a well-known example.\footnote{We postpone its well-known description to Online Appendix \ref{app:da_eada}.}
Let $\da(P)$ denote the matching produced by Deferred Acceptance in problem $P$, and let $\da_i (P)$ denote student $i$'s assignment under $\da(P)$.

The envy digraph $G^{\da(P)}$ has vertex set $I$ and a directed edge $i\rightarrow j$ if student $i$ prefers $\da_j(P)$ to $\da_i(P)$.
A cycle is a sequence $(i_1,\dots,i_k)$ of distinct students with edges $i_\ell\rightarrow i_{\ell+1}$ for $\ell=1,\dots,k-1$ and $i_k\rightarrow i_1$.
A cycle packing $\Pi$ is a collection of disjoint cycles in $G^{\da(P)}$.
Let $V(\Pi)\subseteq I$ be the set of students covered by cycles in $\Pi$.
For $j\in I$, let $N^-(j)\coloneqq\{h\in I:\ h\rightarrow j \text{ in }G^{\da(P)}\}$ denote the in-neighborhood of $j$ in the envy digraph, i.e. the students who envy $j$ in DA.
It is well-known that every cycle packing in $G^{\da(P)}$ corresponds to a sequence of trades that lead to a matching that Pareto dominates $\da(P)$ \citep[e.g.][]{dur2019school}.

\paragraph{Improvable students.}
Given a problem $P$, let $\mathcal{M} (P)$ denote the set of matchings that weakly Pareto-dominate $\da(P)$.
A student $i\in I$ is unimprovable if $\mu_i=\da_i(P)$ for every $\mu\in\mathcal{M}(P)$.
Let $\mathcal{U}(P)$ denote the set of unimprovable students, and let $\mathcal{I}(P)\coloneqq I\setminus \mathcal{U}(P)$ denote the set of improvable students.

The well-known observation below provides a useful connection between improvability and the structure of the DA envy digraph.

\begin{lemma}[\cite{tang2014new, dur2019school}]
	\label{lem:improvable_cycle}
	The following statements are equivalent:
	\begin{enumerate}
		\item A student is improvable.
		\item The corresponding node lies on a cycle in the DA envy digraph $G^{\da(P)}$.
	\end{enumerate}
\end{lemma}

\paragraph{Justifiability}
\label{sec:cc}

We now formalize which priority violations can be justified when implementing efficiency improvements over $\da(P)$.
The idea behind justifiability is that, while any Pareto improvement over $\da(P)$ must violate some priorities, not all violations are equally defensible.
We distinguish between violations affecting students who can never benefit from any improvement (unimprovable students) and those affecting students who benefit from the improvement itself.

For any matching $\mu\in\mathcal{M}(P)$, define the set of students who benefit from $\mu$ relative to DA.\footnote{For simplicity, we use the notation $\BB(\mu)$ instead of the more cumbersome $\BB(\mu, P)$.}

\begin{definition}
	Given $\mu\in\mathcal{M}(P)$, the set of beneficiaries is
	\[
	\BB(\mu)\coloneqq \{i\in I:\rk_i(\mu_i)<\rk_i(\da_i(P))\}.
	\]
\end{definition}

Unimprovable students are unaffected by the choice of efficiency improvement: regardless of which Pareto improvement is selected, they remain assigned to their DA school.
By contrast, improvable non-beneficiaries are students who could benefit under some Pareto improvement but do not benefit under the particular improvement $\mu$.

We now introduce justifiability, our main conceptual contribution: a restriction on which priority violations can be justified when improving upon DA.

\begin{definition}
	Given $\mu\in\mathcal{M}(P)$, a priority violation against student $i$ under $\mu$ is justifiable if
	\[
	i\in \mathcal{U}(P)\ \cup\ \BB(\mu).
	\]
	A priority violation against an improvable non-beneficiary is unjustifiable.
\end{definition}

The criterion is endogenous: whether a violation is justifiable depends on what the improvement accomplishes for the affected student, not on permissions granted in advance.

\begin{definition}
	A matching $\mu\in\mathcal{M}(P)$ is justifiable if every priority violation in $\mu$ is justifiable.
\end{definition}

\josue{The normative principle in our justifiability definition is that, if a student can be compensated for their priority violation, they should receive such compensation in the form of a better-than-DA placement, subject to the restriction of everyone keeping a placement as good as their DA allocation.
	Since this is impossible for unimprovable students, we think that the violation of their priorities is defensible on normative grounds.
	We remark that unimprovable students are the first to be removed in EADA, which is consistent with our view that, if any priorities are to be overridden, theirs should be the first.
	Moreover, while every problem contains at least one unimprovable student, in large random markets they constitute only a small fraction of the student population \citep{ortega2025identifying}.  }

\josue{Improvable beneficiaries are compensated for any potential priority violations.
	Our criterion is deliberately permissible regarding the size of the improvement needed to be for a priority violation to be tolerated.
	But what in our view cannot be normatively be defended is that a priority violation is uncompensated when such compensation is possible, as that would create the potential complaint: \emph{why was my priority overridden while choosing an improvement that gives me nothing, when another improvement could have benefited me?}
	Our criterion rules out precisely this combination of a priority violation and exclusion from feasible gains.
	It does not require every improvable student to benefit, since different students' improvement opportunities may conflict.
	Instead, it requires that the priorities of those left at their DA assignments remain protected whenever they are improvable.}

\paragraph{Labelled Envy Digraph.}
The envy digraph $G^{\da(P)}$ records who envies whose assignment under $\da(P)$ and tells us which improvements can occur.
To analyze the priority violations that such exchanges would generate, we  enrich $G^{\da(P)}$ with labels.
To each edge $i\rightarrow j$ we associate the set of improvable students who desire school $\da_j(P)$ (i.e., point to $j$ in $G^{\da(P)}$) and whose priority at $\da_j(P)$ would be violated if $i$ were assigned to $\da_j(P)$:
\[
l(i\rightarrow j)\coloneqq \{h\in \mathcal{I}(P)\cap N^-(j) : \rk_{\da_j(P)}(h)<\rk_{\da_j(P)}(i)\}.
\]
For a cycle packing $\Pi$, define $l(\Pi)\coloneqq \bigcup_{(i\rightarrow j)\in \Pi} l(i\rightarrow j)$.
We remark that the labels tell us the priorities that an exchange would violate fixing the allocation for every other student to be DA, without verifying the actual cycle implemented: $i \stackrel{k}{\rightarrow} j$ tell us that an improvement cycle where $i$ takes $j$'s place would  \emph{potentially} violate $k$'s priority, without verifying whether $k$'s priority is actually violated in the actual matching generated by the exchange (which may not occur if $k$ is assigned to a school she prefers over $\da_j(P)$).

We now connect justifiability to the labelled DA envy digraph.
Recall that a cycle packing $\Pi$ consists of disjoint cycles in $G^{\da(P)}$ and induces a Pareto improvement over $\da(P)$ by assigning each student in a cycle the DA school of her successor.
The set $l(\Pi)$ collects the improvable students whose priorities may be violated by executing the trades in $\Pi$.
This allows us to express justifiability as a simple condition on labels.

\begin{lemma}
	\label{prop:cc_packing}
	Fix a problem $P$ and a cycle packing $\Pi$ with induced matching $\mu^\Pi$.
	The matching $\mu^\Pi \in \mathcal{M}(P)$ is justifiable if and only if:
	\[
	l(\Pi)\subseteq \BB(\mu^\Pi).
	\]
\end{lemma}

\begin{proof}
	Suppose first that $\mu^\Pi$ is justifiable, and let $h\in l(\Pi)$.
	Then $h\in l(i\to j)$ for some edge $i\to j$ in $\Pi$.
	Thus $h$ is improvable, prefers $\da_j(P)$ to $\da_h(P)$, and has higher priority at $\da_j(P)$ than $i$.
	If $h\notin \BB(\mu^\Pi)$, then $h$ does not belong to any cycle in $\Pi$, so $\mu^\Pi_h=\da_h(P)$.
	Hence assigning $i$ to $\da_j(P)$ creates a priority violation against $h$ under $\mu^\Pi$.
	Since $h$ is improvable and not a beneficiary, this violation is unjustifiable, contradicting that $\mu^\Pi$ is justifiable.
	Therefore $h\in \BB(\mu^\Pi)$, and so $l(\Pi)\subseteq \BB(\mu^\Pi)$.
	
	Conversely, suppose that $l(\Pi)\subseteq \BB(\mu^\Pi)$.
	Consider $\mu^\Pi$ against some student $h$ at a school $\da_j(P)$ assigned to student $i$ via an edge $i\to j$ in $\Pi$.
	If $h\in\UU(P)$ the violation is justifiable by definition, so suppose $h\in\IR(P)$.
	Then $h$ prefers $\da_j(P)$ to $\mu^\Pi_h$, and has higher priority at $\da_j(P)$ than $i$.
	Since $\mu^\Pi$ Pareto-dominates $\da(P)$, we have $\mu^\Pi_h \succcurlyeq_h \da_h(P)$, so in particular $h$ prefers $\da_j(P)$ to $\da_h(P)$.
	Thus $h\in l(i\to j)\subseteq l(\Pi)$.
	By assumption, $h\in \BB(\mu^\Pi)$.
	Hence every priority violation under $\mu^\Pi$ is justifiable, and therefore $\mu^\Pi$ is justifiable.
	
\end{proof}

In particular, note that a matching $\mu^\Pi$ is justifiable if the following local condition on edges holds: $l(\Pi)=\emptyset$.
In such case, we will say that the corresponding matching is \emph{strongly justifiable}.
Strong justification is deliberately conservative: it rules out every potential DA-relative priority violation against an improvable student, even when that student would be compensated elsewhere in the implemented cycle packing.
It is also useful for a technical reason.
To verify that a matching is justifiable, one must verify each label in a cycle against every node who is part of the cycle packing.
On the other hand, this is much easier in a strongly justifiable matching: we only need to check that each label in a cycle is $\emptyset$, and thus the justifiability property can be verified locally rather than globally.

%%%%%%%%%%%%%%%%%%%%%%%%
\section{Results}
\label{sec:results}
%%%%%%%%%%%%%%%%%%%%%%%%

\subsection{Existence and Structure}

We first show that our notion of justifiability is not vacuous, and gives us at least one Pareto-improvement over DA whenever DA's matching is inefficient.
To do so, we describe a simple mechanism that constructs such an improvement using the set of schools who reject at least one improvable student.\\

We will use the following notation.
Let $S^*(P)\subseteq S$ denote the set of schools that reject at least one improvable student during the execution of DA.
For each school $s\in S^*(P)$, the \emph{cutoff student} at $s$, denoted by  $\underline{i}(s)$, is the lowest-priority student assigned to $s$ in DA.
Similarly, for each $s\in S^*(P)$, the \emph{below cutoff set} of $s$ is
\[
A(s)\coloneqq \{i\in \mathcal{I}(P):\ s\succ_i \da_i(P)\ \text{and}\ \underline{i}(s)\triangleright_s i\},
\]
i.e., the improvable students who prefer $s$ to their DA assignment but have lower priority at $s$ than the cutoff.
Note that, since $s\in S^*(P)$, the set $A(s)$ is non-empty.
Thus, for any school $s\in S^*(P)$, we can find the \emph{just below cutoff student} at $s$, denoted $i_s^*$, which is the highest-priority student in $A(s)$ according to $\triangleright_s$.

\noindent {\it The Just-Below Cutoffs (JBC) mechanism.}

\begin{enumerate}
	\item Construct a directed graph on $S^*(P)$ by adding, for each $s\in S^*(P)$, a directed edge $s\rightarrow \da_{i_s^*}(P)$ (recall $i_s^*$ is the just below cutoff student).
	\item Since every node has out-degree exactly one, the digraph contains at least one cycle.
	For every such cycle $(s_1\rightarrow s_2\rightarrow\cdots\rightarrow s_k\rightarrow s_1)$, execute the corresponding trades simultaneously: each student $i_{s_\ell}^*$ moves from $s_{\ell+1}$ to $s_\ell$ (indices mod $k$).
\end{enumerate}

We use $\jbc(P)$ to denote the matching produced by JBC.
We show below that JBC always produces a strongly justifiable improvement whenever DA is Pareto-inefficient.

\begin{theorem}
	\label{thm:existence}
	Every problem with an inefficient DA outcome admits a strongly justifiable improvement.
	
	Moreover, the set of strongly justifiable matchings coincides with the set of matchings obtained by executing subsets of the disjoint cycles found by JBC.
\end{theorem}
We provide only a proof sketch here and postpone the full proof to Appendix \ref{app:proof1}.
The idea is as follows: every over-demanded school has a unique highest priority student who desires it that sits just below the cutoff, because priorities are strict.
In JBC, the graph construction is such that each such school has out-degree one, and by finiteness of the graph, a cycle must exist.
Moreover, the edges in this cycle have empty labels because each entrant is the highest-priority improvable student below the relevant cutoff.
On the other hand, if there was a strongly justifiable matching that did not correspond to a selection of the JBC disjoint cycles, that would mean that some school admits another student other than its just-below-cutoff student.
Because that student has worse priority than the just-below-cutoff student, the corresponding edge would have a non-empty label, contradicting strong justifiability.
So in every strongly justifiable matching, if a school can admit a new student, it can only admit the same student in every strongly justifiable improvement.
Once that student moves, the seat she leaves must also be filled; repeating the same argument forces the completion of the corresponding JBC cycle.

A direct implication of Theorem~\ref{thm:existence} is that every strongly justifiable matching is obtained by executing a subset of the disjoint cycles found by JBC.
Because these cycles do not share any school, one strongly justifiable matching Pareto-dominates another if and only if it executes a superset of the latter’s cycles.
Thus the set of strongly justifiable matchings forms a lattice under the Pareto-dominance order, with JBC as the unique maximal element and \(\da(P)\) as the minimal element.

We now illustrate how JBC operates in Example~\ref{ex:running}.

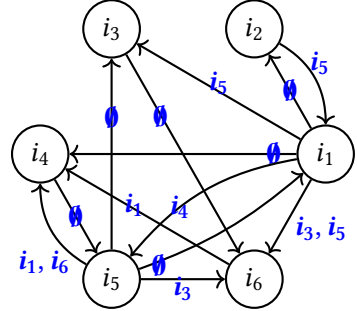
\begin{example}[Problem (above) and labelled envy digraph (below)]
	\label{ex:running}
	\begin{center}
		\begin{minipage}[c]{0.8\textwidth}
			\centering
			
			\scalebox{0.85}{%
				\begin{tabular}{ccccccc|ccccccc}
					\hline
					$i_1$ & $i_2$ & $i_3$ & $i_4$ & $i_5$ & $i_6$ & $i_7$ & $s_1$ & $s_2$ & $s_3$ & $s_4$ & $s_5$ & $s_6$ & $s_7$ \\
					\hline
					$s_6$ & $s_1$ & $s_6$ & $s_5$ & $s_3$ & $s_4$ & $s_4$ &
					$\bm{i_1}$ & $\bm{i_2}$ & $\bm{i_3}$ & $\bm{i_4}$ & $\bm{i_5}$ & $\bm{i_6}$ & $\cdot$ \\
					$s_4$ & $\bm{s_2}$ & $\bm{s_3}$ & $\bm{s_4}$ & $s_6$ & $\pmb{s_6}$ & $\pmb{s_7}$ &
					$i_5$ & $i_1$ & $i_5$ & $i_7$ & $i_4$ & $i_3$ & $\cdot$ \\
					$s_2$ & $\cdot$ & $\cdot$ & $\cdot$ & $s_4$ & $\cdot$ & $\cdot$ &
					$i_2$ & $\cdot$ & $i_1$ & $i_1$ & $i_1$ & $i_5$ & $\cdot$ \\
					$s_3$ & $\cdot$ & $\cdot$ & $\cdot$ & $s_1$ & $\cdot$ & $\cdot$ &
					$\cdot$ & $\cdot$ & $\cdot$ & $i_6$ & $\cdot$ & $i_1$ & $\cdot$ \\
					$s_5$ & $\cdot$ & $\cdot$ & $\cdot$ & $\bm{s_5}$ & $\cdot$ & $\cdot$ &
					$\cdot$ & $\cdot$ & $\cdot$ & $i_5$ & $\cdot$ & $\cdot$ & $\cdot$ \\
					$\bm{s_1}$ & $\cdot$ & $\cdot$ & $\cdot$ & $\cdot$ & $\cdot$ & $\cdot$ &
					$\cdot$ & $\cdot$ & $\cdot$ & $\cdot$ & $\cdot$ & $\cdot$ & $\cdot$ \\
					\hline
			\end{tabular}		}~
			
		\end{minipage}\hfill
		\begin{minipage}[c]{0.8\textwidth}
			\centering
			\small
			
			\begin{tikzpicture}[scale=0.75, ->, node distance=1.5cm,
				every node/.style={draw, circle, minimum size=0.4cm}, thick]
				\def\radius{3.8}
				\def\n{6} % number of nodes
				
				% Nodes in a circle
				\foreach \i/\name in {1/i_1, 2/i_2, 3/i_3, 4/i_4, 5/i_5, 6/i_6} {
					\node (\name) at ({360/\n * (\i-1)}:\radius) {$\name$};
				}
				
				% Edges (your original)
				\draw (i_1) to node[blue,pos=0.7, above,draw=none, yshift=-15pt,xshift=15pt]{$\bm{i_3,i_5}$} (i_6);
				\draw (i_1) to node[blue, pos=0.1, above,draw=none, yshift=-10pt]{$\bm{\emptyset}$} (i_4);
				\draw (i_1) to node[blue,midway, above,draw=none, yshift=-10pt]{$\bm \emptyset$} (i_2);
				\draw (i_1) to node[blue,midway, above,draw=none, yshift=-5pt]{$\bm{i_5}$}(i_3);
				\draw[bend right=20] (i_1) to node[blue,pos=0.66, above,draw=none, yshift=-4pt]{$\bm{i_4}$} (i_5);
				
				\draw[bend left] (i_2) to node[blue,midway, above,draw=none, yshift=-6pt,xshift=2pt]{$\bm{i_5}$} (i_1);
				
				\draw (i_3) to node[blue,pos=0.3, above,draw=none, yshift=-11pt]{$\bm{\emptyset}$} (i_6);
				
				\draw (i_4) to node[blue,pos=0.5,draw=none, yshift=0pt]{$\bm{\emptyset}$} (i_5);
				
				\draw[bend right=10] (i_5) to node[blue,pos=0.1, above,draw=none, yshift=-12pt]{$\bm{\emptyset}$} (i_1);
				\draw[bend left] (i_5) to node[blue,pos=0.3,draw=none, xshift=-7pt, yshift=-6pt]{$\bm{i_1, i_6}$} (i_4);
				
				\draw (i_5) to node[blue,pos=0.5, below,draw=none, yshift=8pt]{$\bm{i_3}$} (i_6);
				\draw (i_5) to node[blue,pos=0.7, above,draw=none, yshift=-11pt]{$\bm{\emptyset}$} (i_3);
				
				\draw (i_6) to node[blue,pos=0.58, below,draw=none, yshift=17pt]{$\bm{i_1}$} (i_4);
			\end{tikzpicture}
			
		\end{minipage}
	\end{center}

\end{example}

For each school $s\in S^*(P)$, JBC identifies the highest-priority student $i_s^*$ who prefers $s$ to her DA assignment.
This defines the directed school graph in Figure~\ref{fig:tpcexample1}.
The unique cycle $(s_1\rightarrow s_5\rightarrow s_4\rightarrow s_1)$ yields the student exchange
\[
(i_1 \stackrel{\emptyset}{\rightarrow} i_4 \stackrel{\emptyset}{\rightarrow} i_5 \stackrel{\emptyset}{\rightarrow} i_1).
\]

\begin{figure}[h!]
	\begin{tikzpicture}[scale=0.4, ->, node distance=2.5cm, every node/.style={draw, circle, minimum size=.7cm}, thick]
		\def\radius{3}
		
		% Define nodes in a circular layout
		\foreach \i/\name in {1/s_1, 2/s_2, 3/s_3, 4/s_4, 5/s_5, 6/s_6} {
			\node[circle, draw] (\name) at ({360/6 * (\i-1)}:\radius) {$\name$}; 
		}
		
		% Color nodes 1-6 with red borders
		
		\draw[ultra thick] (s_1) to (s_5);
		\draw[thick] (s_2) to (s_1);
		\draw[thick] (s_3) to (s_5);
		\draw[ultra thick] (s_4) to (s_1);
		\draw[ultra thick] (s_5) to (s_4);
		\draw[thick] (s_6) to (s_3);
	\end{tikzpicture}
	\caption{School graph induced by JBC in Example~\ref{ex:running}.
		Thick edges form the JBC cycle.}
	\label{fig:tpcexample1}
\end{figure}
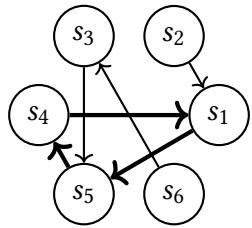

\begin{table}[h!]
	\centering
	\caption{Cycle Packings in Example~\ref{ex:running}}
	\label{tab:improvements_example1}
	\small
	\begin{tabular}{llll}
		\toprule
		\textbf{Cycle Packing} & \textbf{SJ?} & \textbf{J?} & \textbf{PE?} \\
		\midrule
		$(i_1 \stackrel{\emptyset}{\rightarrow} i_2 \stackrel{i_5}{\rightarrow} i_1)$ & No & No & No\\
		$(i_1 \stackrel{i_4}{\rightarrow} i_5 \stackrel{\emptyset}{\rightarrow} i_1)$ & No & No & No\\
		$(i_4 \stackrel{\emptyset}{\rightarrow} i_5 \stackrel{i_1,i_6}{\rightarrow} i_4)$ & No & No & No\\
		$(i_1 \stackrel{i_3, i_5}{\rightarrow} i_6 \stackrel{i_1}{\rightarrow} i_4 \stackrel{\emptyset}{\rightarrow} i_5 \stackrel{\emptyset}{\rightarrow} i_1)$ & No & No & Yes\\
		$(i_3 \stackrel{\emptyset}{\rightarrow} i_6 \stackrel{i_1}{\rightarrow} i_4 \stackrel{\emptyset}{\rightarrow} i_5 \stackrel{\emptyset}{\rightarrow} i_3)$ & No & No & No\\
		$(i_1 \stackrel{\emptyset}{\rightarrow} i_4 \stackrel{\emptyset}{\rightarrow} i_5 \stackrel{\emptyset}{\rightarrow} i_1)$ & Yes & Yes & No\\
		$(i_1 \stackrel{i_5}{\rightarrow} i_3 \stackrel{\emptyset}{\rightarrow} i_6 \stackrel{i_1}{\rightarrow} i_4 \stackrel{\emptyset}{\rightarrow} i_5 \stackrel{\emptyset}{\rightarrow} i_1)$ & No & Yes & No\\
		$\{(i_1 \stackrel{\emptyset}{\rightarrow} i_2 \stackrel{i_5}{\rightarrow} i_1),(i_3 \stackrel{\emptyset}{\rightarrow} i_6 \stackrel{i_1}{\rightarrow} i_4 \stackrel{\emptyset}{\rightarrow} i_5 \stackrel{\emptyset}{\rightarrow} i_3)\}$ & No & Yes & Yes\\
		\bottomrule
		\scriptsize SJ: strongly justifiable; J: justifiable; PE: Pareto-efficient.
	\end{tabular}
\end{table}
\iffalse
\input{difference}
\fi
Example~\ref{ex:running} illustrates that JBC selects a strongly justifiable improvement but the resulting matching need not be Pareto-efficient.
At the same time, justifiability allows for further efficiency gains once we allow for edges that are not necessarily empty-labelled (Table~\ref{tab:improvements_example1}).\footnote{Justifiability does not require preserving every JBC beneficiary; see Example~\ref{ex:notcorebenefi} in Online Appendix~\ref{app:justifiable_not_core}.
	Moreover, the set of justifiable matchings need not form a lattice.
	In Example~\ref{ex:running}, the JBC outcome and the justifiable Pareto-efficient matching displayed in Table~\ref{tab:improvements_example1} have no common Pareto upper bound: any such upper bound would have to weakly Pareto-dominate the Pareto-efficient matching and hence coincide with it, but that matching does not weakly Pareto-dominate the JBC outcome.
	Hence this pair has no join.
	Our main algorithm builds on JBC by deliberately preserving its beneficiaries during the expansion phase.}
In particular, in Example 1 we can find a cycle packing that is both justifiable and Pareto-efficient.
However, this fortunate coincidence does not hold in general, as we show below.

\begin{proposition}
	\label{prop:efficiency}
	There exists a problem $P$ that admits no matching that is both justifiable and Pareto-efficient.
\end{proposition}

\begin{proof}
	Consider the school choice problem in the Example \ref{ex:noeff} below, with six students and six schools with unit capacity.
	The DA outcome assigns each student to the school shown in bold.
	\begin{example}[Preferences and priorities (above); envy digraph (below)]\label{ex:noeff}
		\centering
		
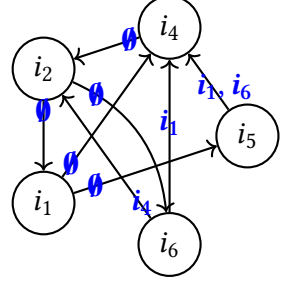
		\begin{minipage}[c]{0.80\textwidth}
			\centering
			\scalebox{0.85}{\begin{tabular}{cccccc|cccccc}
					\hline
					$i_1$ & $i_2$ & $i_3$ & $i_4$ & $i_5$ & $i_6$ 							& $s_{1}$ & $s_2$ & $s_{3}$ & $s_{4}$ & $s_{5}$ & $s_{6}$ \\
					\hline
					$s_{5}$ & $s_{6}$ & $s_{1}$ & $s_2$ & $s_{4}$ & $s_2$ 				& \bm{$i_1$} & \bm{$i_2$} 	& $\cdot$ 	& \bm{$i_4$} & \bm{$i_5$} & \bm{$i_6$} \\
					$s_{4}$ & $s_{1}$ & $s_2$ & \bm{$s_{4}$} & \bm{$s_{5}$} & $s_{4}$ 	& $i_3$ & $i_4$ 			& $\cdot$ 	& $i_3$ & $i_1$ & $i_2$ \\
					\bm{$s_{1}$} & \bm{$s_2$} & $s_{6}$ & {} & {} & \bm{$s_{6}$} 		& $i_2$ & $i_3$ 			& $\cdot$ 	& $i_1$ & 		&  \\
					{} & {} & $s_{4}$ & {} & {} & {}					 					& 		& $i_6$ 			&  			& $i_6$ & 		& \\
					{} & {} & \bm{$s_{3}$} & {} & {} & {} 									& 		& 					&  			& $i_5$ & 		&  \\
			\end{tabular}}
		\end{minipage}~
		
		\begin{minipage}[c]{0.8\textwidth}
			\centering
			\begin{tikzpicture}[scale=0.75, ->, node distance=2.5cm, every node/.style={draw, circle, minimum size=.7cm}, thick]
				\def\radius{3}
				
				% Define nodes in a circular layout
				\foreach \i/\name in {1/i_1, 2/i_2, 3/i_4, 4/i_5, 5/i_6} {
					\node[circle, draw] (\name) at ({-360/5 * (\i+1)}:\radius) {$\name$}; 
				}
				
				% Connections between nodes
				\draw[thick] (i_1) to node[blue,pos=0.1, above,draw=none, yshift=-11pt]{$\bm{\emptyset}$}  (i_4);
				\draw[thick] (i_1) to node[blue,pos=0.15, above,draw=none, yshift=-11pt]{$\bm{\emptyset}$} (i_5);
				
				\draw[thick] (i_2) to node[blue,pos=0.15, above,draw=none, yshift=-11pt]{$\bm{\emptyset}$} (i_1);
				\draw[thick, bend left] (i_2) to node[blue,pos=0.15, above,draw=none, yshift=-11pt]{$\bm{\emptyset}$} (i_6);
				
				\draw[thick] (i_4) to node[blue,pos=0.15, above,draw=none, yshift=-11pt]{$\bm{\emptyset}$} (i_2);
				
				\draw[thick] (i_5) to node[blue,pos=0.1, above,draw=none, yshift=-6pt,xshift=4pt]{$\bm{i_1, i_6}$} (i_4);
				
				\draw[thick] (i_6) to node[blue,pos=0.1, above,draw=none, yshift=-11pt,xshift=-6pt]{$\bm{i_4}$} (i_2);
				\draw[thick] (i_6) to node[blue,pos=0.65, below,draw=none, yshift=9pt,xshift=-4pt]{$\bm{i_1}$} (i_4);
			\end{tikzpicture}
		\end{minipage}
	\end{example}
	
	\noindent {\bf Step 1} [Only justified improvement].
	The labelled DA envy digraph contains exactly four cycles, and all four contain student $i_2$.
	Hence no two cycles are disjoint, so the only non-empty cycle packings are the singleton packings generated by these four cycles.
	Among them, the only cycle whose traded edges are justifiable is
	\[
	(i_1 \stackrel{\emptyset}{\rightarrow} i_4 \stackrel{\emptyset}{\rightarrow} i_2 \stackrel{\emptyset}{\rightarrow} i_1).
	\]
	Indeed, each traded edge on this cycle has empty label, so executing it violates no improvable student's priority, and therefore the induced improvement is (strongly) justifiable.
	
	\noindent By contrast, each of the other three cycles contains at least one traded edge whose label includes an improvable student who is not a beneficiary of that cycle, and thus fails justifiability:
	\begin{itemize}
		\item The cycle $(i_2 \stackrel{\emptyset}{\rightarrow} i_6 \stackrel{i_4}{\rightarrow} i_2)$ is not justifiable because the edge $i_6 \rightarrow i_2$ has label $\{i_4\}$, so executing the cycle violates the priority of $i_4$, who is not improved by that cycle.
		\item The cycle $(i_2 \stackrel{\emptyset}{\rightarrow} i_6 \stackrel{i_1}{\rightarrow} i_4 \stackrel{\emptyset}{\rightarrow} i_2)$ is not justifiable because the edge $i_6 \rightarrow i_4$ has label $\{i_1\}$, so it violates the priority of $i_1$, who is not improved by that cycle.
		\item The cycle $(i_2 \stackrel{\emptyset}{\rightarrow} i_1 \stackrel{\emptyset}{\rightarrow} i_5 \stackrel{i_1,i_6}{\rightarrow} i_4 \stackrel{\emptyset}{\rightarrow} i_2)$ is not justifiable because the edge $i_5 \rightarrow i_4$ has label $\{i_1, i_6\}$, so executing the cycle violates the priority of $i_6$, who is not improved by that cycle.
	\end{itemize}
	Consequently, the unique justifiable Pareto improvement over $\da(P)$ in Example~\ref{ex:noeff} is the one induced by the $3$-cycle above; denote the resulting matching by $\mu^{J}$.
	
	\noindent {\bf Step 2} [The justifiable improvement is not efficient].
	We claim that $\mu^{J}$ is not Pareto-efficient.
	After executing the $3$-cycle, students $i_1$ and $i_5$ both strictly prefer to exchange their assigned schools (as is immediate from the preferences in Example~\ref{ex:noeff}), so there exists a further Pareto improvement over $\mu^{J}$ that assigns $i_1$ to $\mu^{J}(i_5)$ and $i_5$ to $\mu^{J}(i_1)$, leaving all other students unchanged.
	However, any matching that implements this exchange violates the priority of student $i_6$ at school $s_4$ (the school assigned to $i_5$ after the exchange), and $i_6$ is an improvable non-beneficiary of the exchange.
	Hence this further improvement cannot be justifiable.
	
	\noindent Since $\mu^{J}$ is the unique justifiable improvement over $\da(P)$ in Example~\ref{ex:noeff}, and $\mu^{J}$ is not Pareto-efficient, it follows that Example~\ref{ex:noeff} admits no matching that is both justifiable and Pareto-efficient.
\end{proof}

Proposition~\ref{prop:efficiency} shows that justifiability and full Pareto-efficiency are generally incompatible.
When restricted to justifiable improvements, therefore, we must settle for constrained-efficiency.
In the next subsection, we introduce a JBC-related algorithm, SJBC+, that always returns a justifiable matching that satisfies one such notion.

%%%%%%%%%%%%%%%%%%%%%%%%%%%%%%%%%%%%%%%%%%%%%%%%%%%%%%%%%%%%%%

\subsection{The Sequential Just Below Cutoffs Algorithm}
\label{sec:stpcplus}

JBC allows only empty-labelled edges.
But justifiability allows more than this: a priority violation against an improvable student is permitted whenever that student benefits from the realized improvement.
This opens the door to a richer class of improvements, that we will obtain via a sequential implementation of JBC.

First, we represent the problem as a bipartite graph $H^t$, where each improvable student appears once on each side: we denote these nodes by $i^L$ (left) and $i^R$ (right).

For each edge $i\to j$ in $G^{\da(P)}$ with $i,j\in\mathcal I(P)$ and $l(i\to j)\subseteq B^t$, we add an edge from $i^L$ (left) to $j^R$ (right).
We also add a self-loop \(i^L\to i^R\) for every improvable student, representing remaining at the DA assignment.
A perfect matching in \(H^t\) then corresponds to a cycle packing: matched pairs \(i\to j\) with \(i\neq j\) form trading cycles, while self-loops represent students who keep their DA assignment.
We aim to find a perfect matching that maximizes the number of non-loop edges, subject to the constraint that every current beneficiary in $B^t$ remains assigned through a non-loop edge.
The previous iteration's perfect matching satisfies this constraint and remains feasible because the admissible edge set expands as $B^t$ grows.
Equivalently, assign weight one to each admissible non-loop edge and weight zero to each self-loop, and compute a maximum-weight perfect matching after forbidding the self-loops of students in $B^t$.\footnote{This is a standard assignment problem, equivalently a minimum-cost perfect matching problem in a bipartite graph, and is solvable in polynomial time.}

Our proposed algorithm iterates this perfect-matching step in the augmented graph \(H^t\).

\noindent{\it The Sequential Just Below Cutoffs (SJBC+) algorithm.}
\begin{enumerate}
	\item Set \(\mu^0=\jbc(P)\) and \(B^1=\BB(\mu^0)\).
	
	Let \(M^0\) be the perfect matching induced by the cycles executed by \(\jbc(P)\): if student \(i\) receives the DA assignment of student \(j\), include the edge \(i^L\to j^R\); if student \(i\) does not participate in a JBC cycle, include the self-loop \(i^L\to i^R\).
	
	\item For each \(t=1,2,\dots\), given \(B^t\), build the augmented bipartite graph \(H^t\) whose left and right vertex sets are both copies of \(\mathcal I(P)\).
	Add every non-loop edge \(i\to j\) in \(G^{\da(P)}\) satisfying \(l(i\to j)\subseteq B^t\), and add every self-loop \(i\to i\).
	
	Starting from the perfect matching induced by the previous iteration, compute a perfect matching \(M^t\) of \(H^t\) that maximizes the number of non-loop edges, subject to the constraint that no student in \(B^t\) is assigned to her self-loop.\footnote{This constraint is crucial.
		Otherwise, we can find larger improvements, but they need not be justifiable because they may not benefit some current beneficiaries, yet they may violate their priorities.}
	This problem is feasible because the perfect matching from the previous iteration remains feasible when the admissible set expands.
	
	Let \(\mu^t\) be the matching induced by \(M^t\), and set
	\[
	B^{t+1}=\BB(\mu^t).
	\]
	
	If \(B^{t+1}=B^t\), stop the expansion phase and let \(\mu^*=\mu^t\).
	Otherwise continue to iteration \(t+1\).
	
	\item Starting from \(\mu^*\), run the refinement phase restricted to the beneficiary set \(B^*=\BB(\mu^*)\).
	At each refinement iteration, call an edge \(i\to j\), with \(i,j\in B^*\), in the envy digraph at the current matching \emph{admissible} if every improvable student \(h\neq j\) who prefers the school currently assigned to \(j\) to \(\da_h(P)\) and has higher priority than \(i\) at that school belongs to \(B^*\).
	Repeatedly find any directed cycle in \(B^*\) consisting only of admissible edges and execute it.
	Continue until no such cycle remains.
	
	\item Output the final matching \(\mu^+\).
\end{enumerate}

We call any matching produced by this procedure an SJBC+ outcome.\footnote{See Online Appendix~\ref{app:why-plus} for an example showing why the refinement step is necessary.}
SJBC+ can yield different final matchings because different maximizing perfect matchings may be selected in the expansion phase, and also because the plus step may close different cycles.
We illustrate the execution of SJBC+ on Example~\ref{ex:running} (for a graphical version see Figure \ref{fig:bipartite_expansion}).
Step~1 yields the JBC cycle $(i_1\stackrel{\emptyset}{\rightarrow} i_4 \stackrel{\emptyset}{\rightarrow} i_5 \stackrel{\emptyset}{\rightarrow} i_1)$ and $B^1=\{i_1,i_4,i_5\}$, with self-loops on $i_2$, $i_3$ and $i_6$.
In the bipartite graph, we find this improvement by following  $i_1^L \to i_4^R \to i_4^L \to i_5^R \to i_5^L \to i_1 ^R \to i_1^L$.

At $t=1$, the newly added edges with labels contained in $B^1$ are: $i_1 \stackrel{i_5}{\rightarrow} i_3$, $i_1 \stackrel{i_4}{\rightarrow} i_5$, $i_2 \stackrel{i_5}{\rightarrow} i_1$, and $i_6 \stackrel{i_1}{\rightarrow} i_4$.
An alternating cycle in the augmented graph \(H^t\), which alternates between edges not in the current perfect matching and edges in the current perfect matching, starts from the self-loop of \(i_2\), namely at \(i_2^L\), and first traverses the newly admissible edge to \(i_1^R\).
The alternating cycle proceeds as follows:
\begin{enumerate}
	\item $i_2^L \to i_1^R$ (unmatched, newly available),
	\item $i_1^R \to i_5^L$ (matched),
	\item $i_5^L \to i_3^R$ (unmatched),
	\item $i_3^R \to i_3^L$ (matched, self-loop),
	\item $i_3^L \to i_6^R$ (unmatched),
	\item $i_6^R \to i_6^L$ (matched, self-loop),
	\item $i_6^L \to i_4^R$ (unmatched, newly available),
	\item $i_4^R \to i_1^L$ (matched),
	\item $i_1^L \to i_2^R$ (unmatched),
	\item $i_2^R \to i_2^L$ (matched, self-loop).
\end{enumerate}
The alternating cycle returns to \(i_2^L\), where it started.
Flipping along this alternating cycle yields the following cycle packing:
\[
(i_1\stackrel{\emptyset}{\rightarrow} i_2 \stackrel{i_5}{\rightarrow} i_1),\; (i_3\stackrel{\emptyset}{\rightarrow} i_6 \stackrel{i_1}{\rightarrow} i_4 \stackrel{\emptyset}{\rightarrow} i_5 \stackrel{\emptyset}{\rightarrow} i_3)
\]
covering all six improvable students.
Thus $B^* = \{i_1,i_2,i_3,i_4,i_5,i_6\}$, the unique justifiable and Pareto-efficient matching from Table~\ref{tab:improvements_example1}.

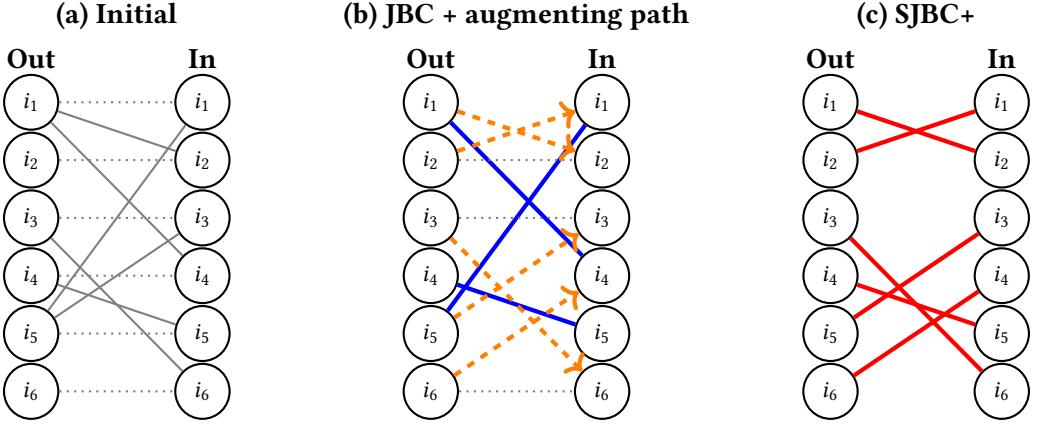
\begin{figure}[h!]
	\centering
	\input{aug_paths}
	\caption{Expansion phase in Example~\ref{ex:running}.
		Panel~(a): after JBC, the JBC cycle (solid blue) covers $i_1$, $i_4$, $i_5$; remaining students are on self-loops (dotted).
		Panel~(b): an alternating cycle (dashed orange) reroutes through the existing perfect matching, incorporating \(i_2\), \(i_3\), \(i_6\).
		Panel~(c): after flipping, all six students are on non-trivial cycles (red).}
	\label{fig:bipartite_expansion}
\end{figure}

Any SJBC+ outcome satisfies several important desiderata.

\begin{theorem}\label{thm:stpcplus}
	If $\da(P)$ is inefficient, any SJBC+ outcome $\mu^+$:
	\begin{enumerate}
		\item Pareto-dominates $\da(P)$;
		\item is justifiable;
		\item \josue{is not Pareto-dominated by any justifiable matching};
		\item satisfies $\BB(\jbc(P)) \subseteq \BB(\mu^+)$
	\end{enumerate}
	Furthermore, any SJBC+ outcome can be computed in time polynomial in $|I|$ and $|S|$.
\end{theorem}

\textcolor{black}{We postpone the full proof to Appendix \ref{app:proofthm2}, but it is clear that the outcome Pareto dominates DA because it executes trades suggested by the envy digraph.
	It is also clear that the outcome is justifiable, as every edge that we add only violates priorities of unimprovable students or beneficiaries from previous iterations, and the set of beneficiaries weakly grows setwise at every step.
	For constrained efficiency, the argument has two parts.
	Any justifiable Pareto dominator would yield either an admissible cycle among the final beneficiaries, contradicting termination of refinement, or a feasible matching in the terminal expansion graph with more beneficiaries, contradicting maximality of the expansion step.
	The set of beneficiaries includes all JBC beneficiaries because JBC is our starting point, at every step weakly increases this set, preserving all members.
	Finally, each expansion step is a standard assignment problem, and the number of expansion and refinement steps is polynomially bounded.}

Before describing SJBC+ performance in simulations, we address two questions.
First, how our justifiability framework is inherently different from existing approaches.
Second, how consent-based mechanisms also face limitations in their ability to restore efficiency if consent is not self-harming.

%%%%%%%%%%%%%%%%%%%%%%%%%%%%%%%%%%%%%%%%%%%%%%%%%%%%
\section{EADA and Ex-ante Consent}
\label{sec:ewp}
%%%%%%%%%%%%%%%%%%%%%%%%%%%%%%%%%%%%%%%%%%%%%%%%%%%%

So far, we have treated consent as an endogenous constraint: whether a priority violation is permissible depends on who ultimately benefits in the realized improvement.
An alternative perspective treats consent as an ex-ante input rather than an ex-post justification.
Under this view, a subset of students agrees in advance to waive priority, and the designer is asked to improve upon the DA outcome while continuing to respect the priorities of all non-consenting students.

We can formalize this perspective through the notion of a consent-based mechanism.
A consent-based mechanism takes as input a school choice problem $P$ and a consent set $W \subseteq I$, interpreted as the set of students whose priorities may be violated, and returns a matching that respects the priorities of $I \setminus W$, i.e. there do not exist $i \in I \setminus W$, $j \in I$, and $s \in S$ such that (i) $i$ prefers $s$ to $\mu_i$, (ii) $\mu_j = s$, and (iii) $\rk_s(i) < \rk_s(j)$.

The celebrated Efficiency-Adjusted Deferred Acceptance mechanism (EADA; \citealp{kesten2010school}) is the leading example of a consent-based mechanism.
For a consent set $W\subseteq I$, let $\eada(P,W)$ denote the matching produced by EADA when the set of consenting students is $W$; we write $\eada_i(P,W)$ for student $i$'s assignment under $\eada(P,W)$.\footnote{EADA with full consent can be implemented by running DA, finding every school who did not reject any student, settling their matching and removing them, and running DA again; iterating this process until all schools are removed.
	A formal description appears in Online Appendix \ref{app:da_eada}.}
It is easy to see that, when only a subset of unimprovable students consent, EADA generates a justifiable matching, since by design EADA only violates priorities of consenting students.
This observation may lead the reader to conclude that our justifiability framework may be strongly linked to EADA's outcomes.
This is not so, as Theorem \ref{prop:cc_vs_eada} shows.

\begin{theorem}
	\label{prop:cc_vs_eada}
	There exists a problem $P$ such that $\eada(P,I)$ is not justifiable, and the unique justifiable and Pareto-efficient matching is not equal to $\eada(P,W)$ for any $W \subseteq I$.
	
\end{theorem}

\begin{proof}
	\noindent [Part 1] In Example~\ref{ex:running}, the EADA outcome under full consent implements the following trades
	\[(i_1 \stackrel{i_3, i_5}{\rightarrow} i_6 \stackrel{i_1}{\rightarrow} i_4 \stackrel{\emptyset}{\rightarrow} i_5 \stackrel{\emptyset}{\rightarrow} i_1)
	\]
	Note that the priority of $i_3$ is violated but she is an improvable student who does not benefit from this exchange.
	Thus, the resulting matching is not justifiable.
	
	\noindent [Part 2] For the same example, in Appendix \ref{app:consenttree} we consider all possible last interrupters during EADA's execution, and construct a consent tree that shows each possible outcome.
	The unique justifiable and Pareto-efficient matching is never produced.
\end{proof}

Figure~\ref{fig:twomatchings} below presents the improvement executed by EADA with full consent, as well as the unique justifiable and Pareto-efficient matching in Example~\ref{ex:running}.\footnote{In this example, the unique justifiable and Pareto-efficient matching improves more students and generates fewer blocking pairs than EADA with full consent; this double-dominance of the EADA outcome is explored by \citet{knipe2025improvable}.}

\begin{figure}[h]
	\centering
	\resizebox{0.9\textwidth}{!}{%  % maintains aspect ratio
		\input{twomatchings}
	}
	\caption{EADA with full consent (left) and the justifiable and efficient matching (right) in Example~\ref{ex:running}.}
	\label{fig:twomatchings}
\end{figure}
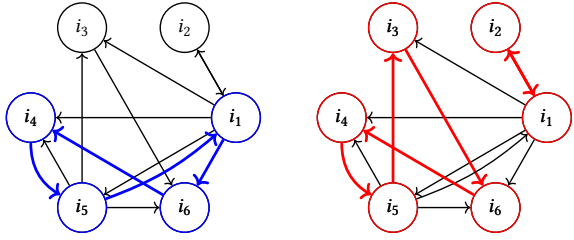
An immediate consequence of Theorem~\ref{prop:cc_vs_eada}, together with the fact that under full consent EADA selects the unique student-optimal priority-efficient assignment and the unique student-optimal legal assignment, is that justifiability is also at odds with other prominent frameworks for controlled priority violations.\footnote{A matching $\mu$ is priority neutral if no matching $\nu$ can make any student whose priority is violated by $\mu$ better off unless $\nu$ violates the priority of some student who is made worse off \citep{reny2022efficient}.
	A matching $\nu$ blocks $\mu$ if there exists a student who prefers her assignment under $\nu$ to her assignment under $\mu$ and who could be admitted to that school ahead of at least one student assigned there under $\mu$.
	A set $L$ of matchings is legal if (i) every matching outside $L$ is blocked by some matching in $L$, and (ii) no matching in $L$ blocks another matching in $L$.
	\citet{ehlers2020legal} show that the legal set is unique and contains a student-optimal element.}

\begin{corollary}
	\label{prop:cc_outside_frameworks}
	The student-optimal priority-efficient assignment and the student-optimal legal assignment need not be justifiable.
	Conversely, there exists a problem $P$ with a unique Pareto-efficient and justifiable matching $\mu$ that is neither legal nor priority-neutral (and hence not priority-efficient).
\end{corollary}

A similar statement can be made regarding essential stability, a weakening of stability that is always satisfied by EADA with full consent \citep{troyan2020essentially}.\footnote{A matching is essentially stable if any priority-based claim initiates a chain of reassignments that results in the initial claimant losing the object.}
It is a direct implication of Theorem \ref{prop:cc_vs_eada} and the aforementioned property that essentially stable matchings need not be justifiable.
One can straightforwardly verify (as we do in Online Appendix \ref{app:essential}) that the unique justifiable and Pareto-efficient matching in Example \ref{ex:running} admits a non-vacuous reassignment chain, showing that a justifiable and Pareto-efficient improvement upon DA need not be essentially stable either.

\subsection{The Limitations of Ex-Ante Consent}

\josue{When all students consent, EADA is an excellent mechanism.
	While its outcome need not be justifiable, EADA achieves Pareto-efficiency and every priority violation in consented ex-ante.
	We think it is fair to say that, under the premise that every student consents, the literature has reached a consensus that EADA is a convincing way to avoid DA's inefficiency on a very appealing normative basis.
	We will show next that such a strong conclusion may not extend when not every student consent, since EADA need not select a Pareto efficient outcome even when there is one available that violates only the consented priorities.
	We will show that this is not a limitation inherent of EADA, but one that extends to every consent-based mechanism.}

To do so, we articulate three properties that one might reasonably want from a consent-based mechanism.
First and foremost, consent should not be self-harming.

\begin{definition}
	\label{def:ic}
	A consent-based mechanism $M$ incentivizes consent if for all $P$, all $W \subseteq I$, and all $i \in I$,
	\[
	\rk_i\!\bigl(M_i(P,\,W \cup \{i\})\bigr)
	\ \le\
	\rk_i\!\bigl(M_i(P,\,W)\bigr).
	\]
\end{definition}

Said differently, every student should be weakly incentivized to consent to waive her priorities.
Second, given a fixed consent set, the mechanism should not leave obvious efficiency gains unexploited.

\begin{definition}
	\label{def:ce}
	A consent-based mechanism $M$ is constrained efficient if for all $P$ and $W \subseteq I$, there is no matching $\mu$ such that
	\begin{enumerate}
		\item $\mu$ respects the priorities of $I \setminus W$, and
		\item $\mu$ Pareto-dominates $M(P,W)$.
	\end{enumerate}
\end{definition}

EADA satisfies both incentivized consent and constrained efficiency, and these guarantees are central to its appeal to both academics and policymakers.\footnote{In fact, EADA satisfies a property stronger than incentivizing consent: giving consent does not harm \emph{or benefit} any consenter \citep{tang2014new}.}
Yet, a different efficiency desideratum is that a consent-based mechanism should return a Pareto-efficient allocation whenever efficiency is compatible with the announced waivers.

\begin{definition}
	\label{def:ewp}
	Given $P$ and $W \subseteq I$, let $\mathcal{M}(P,W)$ denote the set of matchings that
	\begin{enumerate}
		\item respect the priorities of $I \setminus W$,
		\item are Pareto-efficient (among all feasible matchings), and
		\item Pareto-dominate $\da(P)$.
	\end{enumerate}
	A consent-based mechanism $M$ is efficient whenever possible if whenever $\mathcal{M}(P,W) \neq \emptyset$, it selects such a matching, i.e., $M(P,W) \in \mathcal{M}(P,W)$.
\end{definition}

\josue{This requirement extends the full-consent efficiency requirement, satisfied by EADA, to situations in which partial consent already suffices to achieve Pareto-efficiency.
	The motivation for this property is that the whole idea behind consented priority violations is to restore the Pareto-efficiency lost by DA, and if the given priority waivers are enough for us to reach the Pareto frontier, we should be reaching it, as otherwise our efficiency adjustment has failed in an avoidable way.
	Efficiency whenever possible is logically unrelated to constrained efficiency: constrained efficiency only rules out priority-respecting matchings that Pareto-dominate the mechanism's chosen outcome, and may therefore be satisfied by a globally inefficient matching even when a different, incomparable priority-respecting Pareto-efficient improvement exists.
	Conversely, efficiency whenever possible imposes no restriction when no Pareto-efficient DA-improvement is compatible with the announced waivers (which can be the case when only a small fraction of students consent). }

\begin{proposition}
	\label{prop:bigproblem}
	For partial consent sets $W \subsetneq I$, $\eada(P,W)$ is not efficient whenever possible: it may fail to be Pareto-efficient even though $\mathcal{M}(P,W) \neq \emptyset$.
\end{proposition}

\begin{proof}
	Consider Example~\ref{ex:running} and let $W=\{i_1,i_5,i_7\}$.
	Run EADA (see a detailed execution in Online Appendix \ref{app:eada-exec}).
	In the first DA run on $P$, the last interrupter in $W$ is $i_7$ at $s_4$.
	Since $i_7\in W$, EADA deletes $s_4$ from $i_7$'s preference list and reruns DA.
	In the resulting DA run, the only remaining interrupting pair is $(i_3,s_6)$.
	Since $i_3\notin W$, no consenting interrupting pair remains, so the procedure stops.
	Hence $\eada(P,W)$ implements the exchange $(i_1\rightarrow i_4\rightarrow i_5\rightarrow i_1)$, so in particular $\eada_{i_1}(P,W)=s_4$ and $\eada_{i_6}(P,W)=s_6$.
	
	\noindent Define $\nu$ by swapping the assignments of $i_1$ and $i_6$ in $\eada(P,W)$ and leaving everyone else unchanged.
	This yields a feasible matching and makes both $i_1$ and $i_6$ strictly better off (in Example~\ref{ex:running}, $i_1$ prefers $s_6$ to $s_4$ and $i_6$ prefers $s_4$ to $s_6$), so $\nu$ Pareto-dominates $\eada(P,W)$.
	Thus $\eada(P,W)$ is not Pareto-efficient.
	
	\noindent On the other hand, with the same consent set $W$, the cycle packing $\{(i_1\rightarrow i_2\rightarrow i_1),(i_3\rightarrow i_6\rightarrow i_4\rightarrow i_5\rightarrow i_3)\}$ induces a matching that Pareto-dominates $\da(P)$, is Pareto-efficient, and respects the priorities of every non-consenting student in $I\setminus W$ (this can be verified directly from Example~\ref{ex:running}).
	Hence $\mathcal{M}(P,W)\neq\emptyset$.
\end{proof}

Proposition \ref{prop:bigproblem} shows that fixing waivers ex-ante does not, by itself, guarantee that available efficiency gains will be realized.
This raises a deeper question: is the failure specific to EADA, or does it reflect a more fundamental tension among the desiderata above?
Our next Theorem shows that the aforementioned tension is structural in the ex-ante consent framework.

\begin{theorem}
	\label{thm:bigimpo}
	There is no consent-based mechanism $M$ that weakly Pareto-dominates DA and simultaneously satisfies all three properties: (i)~efficiency whenever possible, (ii)~constrained efficiency, and (iii)~incentivizing consent.
\end{theorem}

\begin{proof}
	Fix the school choice problem $P$ from Example~\ref{ex:running}.
	Suppose, by contradiction, that there exists a consent-based mechanism $M$ that weakly Pareto-dominates DA and satisfies efficiency whenever possible, constrained efficiency, and incentivizing consent.
	Consider the following three nested consent sets:
	
	\begin{enumerate}
		\item \textbf{$W_1=\{i_7\}$.}
		In Example~\ref{ex:running}, the unique improvement over $\da(P)$ that can be implemented without violating the priorities of any non-consenting student in $I\setminus W_1$ is the cycle
		\[
		(i_1 \rightarrow i_4 \rightarrow i_5 \rightarrow i_1).
		\]
		This cycle yields a matching that Pareto-dominates $\da(P)$.
		By constrained efficiency, $M(P,W_1)$ must implement this improvement.
		
		\item \textbf{$W_2=\{i_5,i_7\}$.}
		Incentivizing consent implies that $i_5$ must get $s_1$ or better, and must violate the priorities only of $W_2$.
		There is only one improvement that achieves this, i.e. exactly the one we found before:
		\[
		(i_1 \rightarrow i_4 \rightarrow i_5 \rightarrow i_1).
		\]
		
		In particular, the label $\{i_1\}$ on the edge $i_6 \stackrel{i_1}{\rightarrow} i_4$ means that using this edge would violate $i_1$'s priority at $s_4$ unless $i_1$ is herself compensated elsewhere in the implemented packing.
		The 4-node cycle implemented by EADA with full consent provides such compensation, but it is not implementable here because it violates the priority of $i_3\notin W_2$.

		\item \textbf{$W_3=\{i_1,i_5,i_7\}=W_2\cup\{i_1\}$.}
		Incentivizing consent requires that $i_1$ is assigned to $s_4$ or better.
		Only two improvement cycles do this, namely:
		\begin{enumerate}
			\item $(i_1 \stackrel{i_3, i_5}{\rightarrow} i_6 \stackrel{i_1}{\rightarrow} i_4 \stackrel{\emptyset}{\rightarrow} i_5 \stackrel{\emptyset}{\rightarrow} i_1)$
			
			\item	$(i_1 \stackrel{\emptyset}{\rightarrow} i_4 \stackrel{\emptyset}{\rightarrow} i_5 \stackrel{\emptyset}{\rightarrow} i_1)$
			
		\end{enumerate}
		
		Improvement (a) violates the priorities of $i_3 \notin W_3$, so improvement (b) is chosen.
		Yet, the improvement
		\[
		\{(i_1 \stackrel{\emptyset}{\rightarrow} i_2 \stackrel{i_5}{\rightarrow} i_1),(i_3 \stackrel{\emptyset}{\rightarrow} i_6 \stackrel{i_1}{\rightarrow} i_4 \stackrel{\emptyset}{\rightarrow} i_5 \stackrel{\emptyset}{\rightarrow} i_3)\}
		\]
		yields an efficient matching that Pareto-dominates $\da(P)$ while respecting the priorities of all non-consenting students in $I\setminus W_3$.
		Since $M$ is efficient whenever possible, it must satisfy $M(P,W_3)\in \mathcal{M}(P,W_3)$, contradicting the fact that $M(P,W_3)$ is induced by (b) and is inefficient.
		The contradiction completes the proof.
	\end{enumerate}

\end{proof}

Theorem~\ref{thm:bigimpo} identifies an inherent limitation of ex-ante consent designs.
If consent is treated as an arbitrary input, then the natural objective of exploiting all feasible efficiency gains under partial consent conflicts with the very incentive requirement needed to elicit consent in the first place.

Taken together with our earlier results on endogenous justifiability, a coherent picture emerges that shows that achieving efficiency and Pareto-dominating DA comes at important costs, even when some priorities can be violated.
In our endogenous framework, it requires violating the priorities of non-beneficiaries.
On the ex-ante side, even when priority waivers are explicitly granted in advance, robustly achieving efficiency under partial consent runs into an equally sharp obstacle: the mechanism cannot systematically use the available waivers without sometimes undermining the incentives that make waivers available in the first place.

Having clarified these theoretical limits, both under endogenous justifiability and under ex-ante consent, we now turn to the empirical question of how our proposed SJBC+ algorithm fares in practice.
The next section compares SJBC+ and EADA in simulations, shedding light on the practical relevance of the theoretical trade-offs identified above.

%%%%%%%%%%%%%%%%%%%%%%%%%%%%%%%%%%%%%%%%%%%%%%%%%%%
\section{Performance in Practice}
\label{sec:simulations}

\subsection{Simulations}
%%%%%%%%%%%%%%%%%%%%%%%%%%%%%%%%%%%%%%%%%%%%%%%%%%%
The results so far clarify what can and cannot be guaranteed in theory.
We now ask how often these impossibilities bind in random environments, and how frequently the mechanisms nevertheless deliver Pareto-efficient outcomes.
To answer this, we evaluate SJBC+ in simulations and benchmark it against DA and the leading consent-based mechanism, EADA.
To calibrate the ex-ante consent environment, we use observed behavior from laboratory experiments: \citet{cerrone2022school} document that 52\% of subjects voluntarily consent to waive priorities.
We therefore take this 52\% rate as our baseline for generating consent sets in EADA.\footnote{The replication code is available at \url{https://www.josueortega.com}.
	See also \citet{freer2026experimental} for further experimental evidence on EADA's performance under the full consent assumption.}
\josue{Because neither EADA nor SJBC+ is strategy-proof, these comparisons are truthful-report counterfactuals rather than equilibrium predictions.}

We generate random school choice problems with $n \in \{50, 100, 200\}$ students and $n$ schools, each with unit capacity.
For each market size, we construct 2,000 independent instances under two preference structures: (i) independent preferences (students and schools rank uniformly at random), and (ii) correlated preferences, where students exhibit positive correlation in how they rank schools, reflecting realistic settings where school quality is commonly valued.\footnote{Specifically, we draw a common latent value $q_s \sim \mathcal{N}(0,1)$ for each school.
	Each student $i$ ranks schools according to utilities $u_{is} = \rho q_s + \sqrt{1-\rho^2}\,\varepsilon_{is}$, where $\varepsilon_{is} \sim \mathcal{N}(0,1)$ are i.i.d.\ shocks.
	Preferences are obtained by sorting these utilities.
	This specification yields positively correlated rankings across students, with correlation increasing in $\rho$ (here $\rho=0.50$).}
For each instance, we compute outcomes under DA (baseline), SJBC$+$, EADA with 52\% consent, and EADA with full consent.
We measure average student rank, the number of beneficiaries relative to DA, and the frequency with which the outcome is Pareto-efficient or justifiable.

\begin{table}[p]
	\centering
	\caption{SJBC$+$ and EADA in random markets}
	\label{tab:simulation}
	\footnotesize
	\setlength{\tabcolsep}{4pt}
	\renewcommand{\arraystretch}{1.05}
	\newcommand{\simcell}[2]{%
		\shortstack{#1\\{\scriptsize (#2)}}%
	}
	
	\begin{tabularx}{\textwidth}{@{}l*{6}{>{\centering\arraybackslash}X}@{}}
		\toprule
		& \multicolumn{3}{c}{\shortstack{Independent preferences}}
		& \multicolumn{3}{c}{\shortstack{Correlated preferences}} \\
		\cmidrule(lr){2-4}\cmidrule(lr){5-7}
		& $n=50$ & $n=100$ & $n=200$
		& $n=50$ & $n=100$ & $n=200$ \\
		\midrule
		
		\multicolumn{7}{@{}l}{\textbf{Average student rank}} \\
		\quad DA
		& \simcell{4.2}{0.023}
		& \simcell{4.9}{0.025}
		& \simcell{5.7}{0.027}
		& \simcell{10.4}{0.052}
		& \simcell{18.0}{0.088}
		& \simcell{30.9}{0.154} \\
		\quad EADA (full)
		& \simcell{2.6}{0.007}
		& \simcell{2.7}{0.005}
		& \simcell{2.8}{0.004}
		& \simcell{5.3}{0.018}
		& \simcell{6.8}{0.019}
		& \simcell{8.6}{0.020} \\
		\quad EADA (52\%)
		& \simcell{3.3}{0.015}
		& \simcell{3.6}{0.014}
		& \simcell{3.8}{0.013}
		& \simcell{8.1}{0.046}
		& \simcell{12.5}{0.071}
		& \simcell{19.1}{0.110} \\
		\quad SJBC$+$
		& \simcell{2.7}{0.009}
		& \simcell{2.9}{0.008}
		& \simcell{3.0}{0.007}
		& \simcell{5.8}{0.023}
		& \simcell{8.0}{0.029}
		& \simcell{10.9}{0.033} \\
		
		\addlinespace[0.5ex]
		\multicolumn{7}{@{}l}{\textbf{Beneficiaries relative to DA}} \\
		\quad EADA (full)
		& \simcell{19.8}{0.172}
		& \simcell{47.5}{0.275}
		& \simcell{108.3}{0.453}
		& \simcell{32.7}{0.131}
		& \simcell{78.0}{0.165}
		& \simcell{172.0}{0.211} \\
		\quad EADA (52\%)
		& \simcell{11.0}{0.180}
		& \simcell{28.2}{0.322}
		& \simcell{69.1}{0.558}
		& \simcell{14.2}{0.201}
		& \simcell{37.7}{0.370}
		& \simcell{91.4}{0.661} \\
		\quad SJBC$+$
		& \simcell{22.0}{0.240}
		& \simcell{55.6}{0.452}
		& \simcell{132.9}{0.752}
		& \simcell{38.0}{0.217}
		& \simcell{89.9}{0.253}
		& \simcell{191.2}{0.224} \\
		
		\addlinespace[0.5ex]
		\multicolumn{7}{@{}l}{\textbf{Improvable students (\%)}} \\
		\quad Share of all students
		& \simcell{56.8}{0.412}
		& \simcell{68.3}{0.259}
		& \simcell{75.2}{0.174}
		& \simcell{84.9}{0.193}
		& \simcell{92.7}{0.084}
		& \simcell{96.4}{0.041} \\
		
		\addlinespace[0.5ex]
		\multicolumn{7}{@{}l}{\textbf{Beneficiaries among improvable students (\%)}} \\
		\quad EADA (full)
		& \simcell{69.2}{0.309}
		& \simcell{68.8}{0.208}
		& \simcell{71.5}{0.167}
		& \simcell{76.6}{0.203}
		& \simcell{83.9}{0.130}
		& \simcell{89.2}{0.085} \\
		\quad EADA (52\%)
		& \simcell{36.6}{0.508}
		& \simcell{40.0}{0.391}
		& \simcell{45.2}{0.305}
		& \simcell{33.0}{0.445}
		& \simcell{40.5}{0.387}
		& \simcell{47.3}{0.336} \\
		\quad SJBC$+$
		& \simcell{75.4}{0.506}
		& \simcell{79.3}{0.477}
		& \simcell{87.3}{0.388}
		& \simcell{88.9}{0.431}
		& \simcell{96.8}{0.248}
		& \simcell{99.2}{0.105} \\
		
		\addlinespace[0.5ex]
		\multicolumn{7}{@{}l}{\textbf{Priority violations}} \\
		\quad EADA (full)
		& \simcell{8.9}{0.107}
		& \simcell{18.3}{0.158}
		& \simcell{37.0}{0.233}
		& \simcell{42.4}{0.378}
		& \simcell{125.0}{0.802}
		& \simcell{353.8}{1.670} \\
		\quad EADA (52\%)
		& \simcell{3.4}{0.060}
		& \simcell{7.0}{0.095}
		& \simcell{14.0}{0.137}
		& \simcell{6.1}{0.112}
		& \simcell{13.5}{0.186}
		& \simcell{27.7}{0.290} \\
		\quad SJBC$+$
		& \simcell{14.4}{0.245}
		& \simcell{40.9}{0.553}
		& \simcell{109.6}{1.022}
		& \simcell{79.5}{0.858}
		& \simcell{290.2}{1.838}
		& \simcell{906.9}{3.843} \\
		
		\addlinespace[0.5ex]
		\multicolumn{7}{@{}l}{\textbf{Justifiable priority violations (pooled, \%)}} \\
		\quad EADA (full)
		& \simcell{76.4}{0.423}
		& \simcell{69.8}{0.343}
		& \simcell{68.5}{0.257}
		& \simcell{77.7}{0.291}
		& \simcell{80.5}{0.214}
		& \simcell{84.0}{0.165} \\
		\quad EADA (52\%)
		& \simcell{52.0}{0.900}
		& \simcell{44.4}{0.688}
		& \simcell{41.2}{0.480}
		& \simcell{52.8}{0.841}
		& \simcell{50.9}{0.692}
		& \simcell{50.7}{0.553} \\
		\quad SJBC$+$
		& \simcell{100.0}{0.000}
		& \simcell{100.0}{0.000}
		& \simcell{100.0}{0.000}
		& \simcell{100.0}{0.000}
		& \simcell{100.0}{0.000}
		& \simcell{100.0}{0.000} \\

		\addlinespace[0.5ex]
		\multicolumn{7}{@{}l}{\textbf{Justifiable outcomes (\%)}} \\
		\quad EADA (full)
		& \simcell{27.3}{0.996}
		& \simcell{3.2}{0.397}
		& \simcell{0.1}{0.087}
		& \simcell{2.1}{0.324}
		& \simcell{0.2}{0.112}
		& \simcell{0.0}{0.000} \\
		\quad EADA (52\%)
		& \simcell{35.3}{1.069}
		& \simcell{9.7}{0.662}
		& \simcell{0.7}{0.180}
		& \simcell{19.6}{0.888}
		& \simcell{3.0}{0.382}
		& \simcell{0.2}{0.100} \\
		\quad SJBC$+$
		& \simcell{100.0}{0.000}
		& \simcell{100.0}{0.000}
		& \simcell{100.0}{0.000}
		& \simcell{100.0}{0.000}
		& \simcell{100.0}{0.000}
		& \simcell{100.0}{0.000} \\
		
		\addlinespace[0.5ex]
		\multicolumn{7}{@{}l}{\textbf{Pareto-efficient outcomes (\%)}} \\
		\quad EADA (full)
		& \simcell{100.0}{0.000}
		& \simcell{100.0}{0.000}
		& \simcell{100.0}{0.000}
		& \simcell{100.0}{0.000}
		& \simcell{100.0}{0.000}
		& \simcell{100.0}{0.000} \\
		\quad EADA (52\%)
		& \simcell{8.8}{0.634}
		& \simcell{1.0}{0.223}
		& \simcell{0.0}{0.000}
		& \simcell{0.0}{0.000}
		& \simcell{0.0}{0.000}
		& \simcell{0.0}{0.000} \\
		\quad SJBC$+$
		& \simcell{66.9}{1.052}
		& \simcell{62.6}{1.082}
		& \simcell{70.8}{1.017}
		& \simcell{70.5}{1.019}
		& \simcell{85.2}{0.794}
		& \simcell{91.3}{0.629} \\
		
		\bottomrule
	\end{tabularx}
	
	\par\vspace{0.5em}
	\begin{minipage}{\textwidth}
		\scriptsize
		\textit{Notes:} Each column reports 2{,}000 replications with $n$ students and $n$ schools of unit capacity.
		Priorities are independent and uniformly random.
	\end{minipage}
\end{table}

\textcolor{black}{Table~\ref{tab:simulation} reveals three robust results.
	First, SJBC+ attains Pareto-efficiency in the majority of instances (62.60--91.35\%).
	Thus, justifiability is compatible with full Pareto-efficiency in the majority of these random markets.
	Second, SJBC+ compares favorably to EADA under partial consent.
	With 52\% consent, EADA's Pareto-efficiency rate is below 10\% in every cell, and no Pareto-efficient outcome is observed in the 2,000 replications of any correlated cell.
	SJBC+ continues to return Pareto-efficient outcomes in over 70\% of correlated instances.
	SJBC+ also yields more beneficiaries on average than full-consent EADA, while achieving a worse average rank and generating more priority violations.}

\textcolor{black}{Third, SJBC+ is justifiable in every instance, whereas EADA's whole-outcome justifiability rate deteriorates sharply with market size.
	In five of the six cells, fewer than 4\% of full-consent EADA outcomes are justifiable.
	This differs from the share of individual violations that are justifiable, which ranges from 68.5\% to 84.0\% under full consent.
	With 52\% consent, EADA improves fewer students and is much less often Pareto-efficient, but is more often justifiable.}

%%%%%%%%%%%%%%%%%%%%%%%%%%%%%%%%%%%%%%%%%%%%%%%%%%%%%%%%%%%%
\subsection{\textcolor{black}{Chilean data}}
\label{subsec:chile}
%%%%%%%%%%%%%%%%%%%%%%%%%%%%%%%%%%%%%%%%%%%%%%%%%%%%%%%%%%%%

\textcolor{black}{The simulations above were conducted under restricted assumptions about preferences and priorities, and a reasonable question is whether one would obtain similar results in more realistic markets.
	We tackle this question by doing the same exercise using the 2025 data from Chile's centralized School Admission System (SAE), corresponding to applications for the 2026 school year.
	The SAE is a nationwide assignment system in which families submit rank-order lists and schools rank applicants using legally defined priorities and lottery tie-breaking, and an allocation is obtained using DA \citep{correa2019school}.
	The data is publicly available at \url{https://datos.cochid.cl/dataset/mineduc_sistema_de_admision_escolar_sae}.}

\textcolor{black}{For our comparison exercise, we exclude special categories of reserved seats and treat linked family applications as individual applications.
	We perform this construction separately for primary and secondary school entrance covered by the SAE (years 1 and 9 in the data).
	The resulting sample contains 64,893 applicants and 5,496 occupied school-grade programs.
	Online Appendix~\ref{app:chile} describes the data in detail.}

\textcolor{black}{Table~\ref{tab:chile} compares DA, SJBC+, and EADA using the same criteria as in the previous subsection, separately for primary and secondary school admission.}

\begin{table}[htbp!]
	\centering
	\caption{Chilean SAE}
	\label{tab:chile}
	\small
	\begin{tabular}{@{}lrr@{}}
		\toprule
		& \multicolumn{2}{c}{Entry cohort} \\
		\cmidrule(lr){2-3}
		& $1^{\circ}$ básico & $1^{\circ}$ medio \\
		\midrule
		
		\multicolumn{3}{@{}l}{\textbf{Average rank}} \\[-0.25ex]
		\quad DA
		& 1.55 & 2.16 \\
		
		\quad EADA (full)
		& 1.51 & 2.00 \\
		
		\quad EADA (52\%)
		& 1.54 & 2.16 \\
		
		\quad SJBC$+$
		& 1.53 & 2.14 \\[0.6ex]
		
		\multicolumn{3}{@{}l}{\textbf{Beneficiaries (\%)}} \\[-0.25ex]
		\quad EADA (full)
		& 720 \;(2.38) & 2,897 \;(8.36) \\
		
		\quad EADA (52\%)
		& 47 \;(0.16) & 144 \;(0.41) \\
		
		\quad SJBC$+$
		& 349 \;(1.15) & 412 \;(1.19) \\[0.6ex]
		
		\multicolumn{3}{@{}l}{\textbf{Priority violations}} \\[-0.25ex]
		\quad EADA (full)
		& 2,597 & 9,262 \\
		
		\quad EADA (52\%)
		& 39 & 98 \\
		
		\quad SJBC$+$
		& 1,020 & 569 \\[0.6ex]
		
		\multicolumn{3}{@{}l}{\textbf{Justifiable violations (\%)}} \\[-0.25ex]
		\quad EADA (full)
		& 2,229 \;(85.8) & 5,233 \;(56.5) \\
		
		\quad EADA (52\%)
		& 32 \;(82.1) & 57 \;(57.8) \\
		
		\quad SJBC$+$
		& 1,020 \;(100) & 569 \;(100) \\
		
		\bottomrule
	\end{tabular}
	\vspace{0.5em}
	\begin{minipage}{\textwidth}
		\footnotesize
		Notes: The sample contains 30,227 applicants entering primary school ($1^{\circ}$ básico) and 34,666 applicants entering secondary school ($1^{\circ}$ medio) in the 2025 SAE regular round for admission in 2026.
		Percentages in parentheses.
		EADA (52\%) averages 50 independent consent draws; reported counts are rounded to the nearest integer.
		Monte Carlo standard errors are small: at most 0.00020 for average rank, 3.71 beneficiaries, 2.67 priority violations, and 1.86 justifiable violations.
	\end{minipage}
\end{table}
\textcolor{black}{We find that EADA with 52\% consent is the mechanism with fewest improvements but also fewest priority violations: it benefits 47 students in primary school and 144 in secondary school, while generating only 39 and 98 priority violations.
	SJBC$+$ violates substantially more priorities, but does so only when those violations satisfy our justifiability criterion.
	It benefits 349 students in primary school and 412 in secondary school, approximately 7.4 and 2.9 times as many as EADA with 52\% consent, and improves average rank in both cohorts.
	Full-consent EADA benefits even more students: 720 students in primary school and 2,897 in secondary school, attaining the lowest average rank in both cohorts.
	These gains come with substantially more priority violations: 2,597 in primary school and 9,262 in secondary school, of which 86\% and 57\% are justifiable, respectively.
	We do not view this comparison as a failure of either approach.
	In comparison with our comparison exercise using simulations, we see that the robust finding is that SJBC+ is significantly more efficient than EADA with 52\% consent, albeit SJBC+ only generates more improvements than EADA with full consent in the synthetic data but not in the Chilean one.
	Another robust finding is that in either version of EADA a significant fraction of priority violations are not justifiable (slightly over 40\% in the Chilean secondary school data, in the simulations around 15\% to 60\%; this observation seems more prominent with EADA with partial consent actually). }

%%%%%%%%%%%%%%%%%%%%%%%%%%%%%%%%%%%
\section{\textcolor{black}{Incentives}}
\label{sec:incentives}
%%%%%%%%%%%%%%%%%%%%%%%%%%%%%%%%%
\textcolor{black}{Throughout the paper, we have assumed that students (or their parents on their behalf) report preferences truthfully, an assumption which may not hold in practice since every mechanism that Pareto-dominates DA is manipulable.
	We will now argue that this assumption makes sense because: i) manipulations are risky, and ii) JBC and SJBC+ satisfy a weak version of strategy-proofness studied in the literature called top-manipulation proofness. }

\textcolor{black}{From now on, fix every element of school choice problem $P$ besides $\succ$.
	For a mechanism $M$, a preference profile $(\succ)=(\succ_i)_{i \in I}$, and a school $s \succcurlyeq_i M_i(P)$, let $\succ_i^s$ denote the report obtained from $\succ_i$ by moving the (arbitrarily chosen) school $s$ to the top of the list while preserving the relative order of all other schools.
	Denote by $\succ_{-i}$ the truthful preference profile by all other students.
	We now can define the following weak version of strategy-proofness (this property was introduced by \citet{afacan2022improving}, and is satisfied by EADA with full consent).}

\begin{definition}
	\textcolor{black}{	A mechanism $M$ is top-manipulation proof if, for every $\succ$, every student $i$, and every $s \succcurlyeq_i M_i(\succ)$, $M_i(\succ) \succcurlyeq_i M_i(\succ_i^s,\succ_{-i})$.}
\end{definition}

\textcolor{black}{Thus, if a student prefers a school to the assignment she obtains, she cannot profit merely by ranking that school as her first choice.
	This property is considerably weaker than strategy-proofness, but it rules out a particularly natural class of manipulations that are common in the Immediate Acceptance mechanism.
	While JBC and any fixed selection of SJBC+ are not strategy-proof, they all satisfy top-manipulation proofness.\footnote{By fixed, we mean that whenever several perfect matchings or admissible refinement cycles are available, SJBC+ selects among them according to an ex-ante fixed ordering that does not depend on reported rank positions.}}

\begin{proposition}
	\label{prop:top-manipulation}
	\textcolor{black}{JBC and any fixed selection of SJBC+ are top-manipulation-proof.
		Moreover, if $s\succ_i \da_i(P)$, then
		\[
		\da(P')=\da(P)
		\qquad\text{and}\qquad
		\jbc(P')=\jbc(P),
		\]
		where $P'$ denotes the problem obtained from $P$ by replacing $\succ_i$ with $\succ_i^s$.}
\end{proposition}

\josue{We postpone the proof to Appendix \ref{app:top}.
	The idea is simple: putting a desirable school to the top does not change the student's DA assignment, because it does not change the set of schools she prefers to that assignment.
	JBC therefore remains unchanged; and under SJBC$+$, the only possible divergence is that she reaches the promoted school, which cannot leave her strictly better than her truthful final assignment.}

\textcolor{black}{Top-manipulation proofness does not imply strategy-proofness. Successful manipulation work moving genuinely worse schools above the DA allocation, obtaining a undesirable by himself but desirable by others school that allows him to trade in JBC.
	Note that this manipulation is risky, for a student may not know for certain if he would be able to trade an undesirable-by-him school later on, and may end up worse off.
	Said otherwise, JBC and any fixed selection of SJBC+ are not obviously manipulable \citep{troyan2020obvious}.}

\textcolor{black}{On the other hand, it is easy to see that JBC does not satisfy the stronger notion of upper-manipulation proofness (as we show in Online Appendix \ref{app:upper}) which is satisfied by EADA with full consent \citep{afacan2022improving}. On a similar note, truthful reporting need not constitute a Nash equilibrium of the corresponding preference revelation game under complete information.
	Whether these strategic possibilities materially affect the performance of JBC and SJBC+ in practice remains an open question.}

%%%%%%%%%%%%%%%%%%%%%%%%%%%%%%%%%%%%%%%%%%%
\section{Conclusion}
\label{sec:conclusion}
We introduced justifiability as a framework for determining which improvements over DA can be implemented through controlled priority violations, and showed that it enables Pareto improvements that are unattainable under existing consent-based approaches.
We proposed an algorithm that selects the unique maximal strongly justifiable matching, and then built on it to obtain a justifiable improvement that is not Pareto-dominated by any other justifiable matching.
Simulations indicate that, despite the inherent tension between justifiability and full Pareto-efficiency, our proposed algorithm frequently achieves Pareto-efficient outcomes while restricting priority violations to the justifiable class, and it improves the DA assignments of substantially more students than leading consent-based mechanisms under partial consent documented in practice.

\singlespacing      % Single spacing
\setlength{\parskip}{-0.2em} 
% Bibliography
\bibliographystyle{te}
\bibliography{bibliogr}

% Appendix

\newpage
\appendix
\addtocontents{toc}{\protect\setcounter{tocdepth}{-1}}

\section{Proof of Theorem \ref{thm:existence}}
\label{app:proof1}
\begin{proof}
Assume \(\da(P)\) is Pareto-inefficient.
Then \(\mathcal I(P)\neq \emptyset\), and therefore \(S^*(P)\neq \emptyset\).
For each \(s\in S^*(P)\), the set \(A(s)\) is non-empty by definition, so \(i_s^*\) is well-defined.

We first show that the school graph used by JBC is well-defined.
Fix \(s\in S^*(P)\).
Since \(i_s^*\in \mathcal I(P)\), Lemma~\ref{lem:improvable_cycle} implies that \(i_s^*\) belongs to a cycle in \(G^{\da(P)}\).
Let \(h\) be the predecessor of \(i_s^*\) on such a cycle.
Then \(h\) prefers \(\da_{i_s^*}(P)\) to \(\da_h(P)\), so during the execution of DA student \(h\) must have applied to school \(\da_{i_s^*}(P)\) and been rejected there.
Since \(h\) is improvable, it follows that \(\da_{i_s^*}(P)\in S^*(P)\).

Thus the map $f:S^*(P)\to S^*(P)$, given by $f(s)\coloneqq \da_{i_s^*}(P)$, is well-defined.
Because \(S^*(P)\) is finite and each school in \(S^*(P)\) points to exactly one school under \(f\), the graph generated by \(f\) contains disjoint cycles.
JBC executes all such cycles simultaneously: for each cycle
\[
(s_1\rightarrow s_2\rightarrow \cdots \rightarrow s_k\rightarrow s_1),
\]
student \(i_{s_\ell}^*\) moves from \(s_{\ell+1}\) to \(s_\ell\) for every \(\ell=1,\dots,k\) (indices modulo \(k\)).
This new matching is feasible because, along each cycle, every school loses exactly one student and gains exactly one student, and different cycles do not share any school.
It is also Pareto-improving, since \(i_{s_\ell}^*\in A(s_\ell)\) implies that \(i_{s_\ell}^*\) strictly prefers \(s_\ell\) to her DA assignment \(s_{\ell+1}=\da_{i_{s_\ell}^*}(P)\), while all other students keep their DA assignments.

We next show that the matching produced by JBC is strongly justifiable.
Fix a school \(s\) lying on a cycle of the JBC school graph, and let \(j\) be the student assigned to \(s\) under DA who leaves \(s\) when this cycle is executed.
Then \(\da_j(P)=s\), and student \(i_s^*\) prefers \(s\) to her DA assignment.
Hence the JBC trade at school \(s\) corresponds to the edge \(i_s^*\to j\) in \(G^{\da(P)}\).
The label of this edge is
\[
l(i_s^*\rightarrow j)=\{h\in \mathcal I(P)\cap N^-(j): \rk_s(h)<\rk_s(i_s^*)\}.
\]
We claim that \(\mathcal I(P)\cap N^-(j)= A(s)\).
Every student in \(A(s)\) is improvable, prefers \(s=\da_j(P)\) to her DA assignment, and has lower priority at \(s\) than \(j\), so she envies \(j\) and belongs to \(\mathcal I(P)\cap N^-(j)\).
Conversely, take any \(h\in \mathcal I(P)\cap N^-(j)\).
Then \(h\) prefers \(s=\da_j(P)\) to \(\da_h(P)\).
By stability of \(\da(P)\), school \(s\) is full and every student assigned to \(s\) has higher priority at \(s\) than \(h\); in particular \(\rk_s(h)>\rk_s(\underline{\imath}(s))\), so \(h\in A(s)\).
Hence
\[
l(i_s^*\rightarrow j)=\{h\in A(s): \rk_s(h)<\rk_s(i_s^*)\}.
\]
Since \(i_s^*\) is the highest-priority student in \(A(s)\), this set is empty.
So every edge used by JBC is empty-labelled, and the matching produced by JBC is strongly justifiable.
This proves the first statement of the theorem.

For the second statement, let \(\mu^\Pi\) be any strongly justifiable matching induced by a cycle packing \(\Pi\) in \(G^{\da(P)}\).
Since \(\mu^\Pi\) is strongly justifiable, every edge in \(\Pi\) is empty-labelled.
Take any school \(s\) that is entered by some student under \(\mu^\Pi\), and let \(i\) be the student who enters \(s\).
Since \(\mu^\Pi\) is induced by a cycle packing in \(G^{\da(P)}\), there exists some student \(j\) assigned to \(s\) under DA such that \(i\to j\) is an edge of \(\Pi\).

We claim that \(i=i_s^*\).
First, \(i\in A(s)\): because \(i\to j\) is an edge of \(G^{\da(P)}\), student \(i\) is improvable and prefers \(s=\da_j(P)\) to \(\da_i(P)\).
Moreover, since \(j\) is assigned to \(s\) under DA and DA is stable, every student assigned to \(s\) under DA has higher priority at \(s\) than \(i\); in particular \(\underline{\imath}(s)\triangleright_s i\).
Hence \(i\in A(s)\).

Next, let \(h\in A(s)\).
Then \(h\in \mathcal I(P)\) and \(h\) prefers \(s\) to \(\da_h(P)\).
Since \(j\) is assigned to \(s\) under DA, we have \(\da_j(P)=s\), and therefore \(h\) envies \(j\).
Thus \(h\in \mathcal I(P)\cap N^-(j)\).
We have shown that
\[
A(s)\subseteq \mathcal I(P)\cap N^-(j).
\]

Since the edge \(i\to j\) is empty-labelled,
\[
l(i\to j)=\{h\in \mathcal I(P)\cap N^-(j): \rk_s(h)<\rk_s(i)\}=\emptyset.
\]
Therefore there is no student in \(A(s)\) with higher priority than \(i\) at \(s\).
Since \(i\in A(s)\), it follows that \(i\) is the highest-priority student in \(A(s)\), namely \(i_s^*\).
So in any strongly justifiable matching, the student who enters school \(s\) must be \(i_s^*\).

Now suppose \(i_s^*\) enters \(s\) under \(\mu^\Pi\).
Then \(i_s^*\) leaves her DA school \(\da_{i_s^*}(P)=f(s)\).
Since \(\mu^\Pi\) is obtained from a cycle packing, the seat vacated at \(f(s)\) must in turn be filled by some student who enters \(f(s)\) along an empty-labelled edge.
By the previous paragraph, that student must be \(i_{f(s)}^*\).
Repeating the same argument, whenever a school \(s\) is used in a strongly justifiable cycle packing, the school \(f(s)\) must also be used.
Since only finitely many schools are involved, this process must run along a cycle of the graph generated by \(f\).
It follows that every strongly justifiable cycle packing is obtained by selecting some of the cycles found by JBC and executing the corresponding trades.
Conversely, any subset of the cycles found by JBC yields a cycle packing all of whose edges are empty-labelled, and hence induces a strongly justifiable matching.
Therefore the set of strongly justifiable matchings coincides with the set of matchings obtained by executing subsets of the disjoint cycles found by JBC.
\end{proof}
\newpage

\section{Proof of Theorem \ref{thm:stpcplus}}
\label{app:proofthm2}

\begin{proof}
\noindent [Part 1] By Theorem~\ref{thm:existence}, \(\jbc(P)\) Pareto-dominates \(\da(P)\).
At each expansion step, the previous perfect matching remains feasible because admissible edges only expand as \(B^t\) grows, and because \(B^t=\BB(\mu^{t-1})\) implies that every student in \(B^t\) was already matched through a non-loop edge in the previous perfect matching.
The optimization also forbids every student in \(B^t\) from using her self-loop.
Hence \(B^t\subseteq B^{t+1}\).
Moreover, every non-loop edge \(i\to j\) is an edge of \(G^{\da(P)}\), so student \(i\) strictly prefers \(\da_j(P)\) to \(\da_i(P)\).
Therefore every student in \(B^{t+1}\) strictly prefers her assignment under \(\mu^t\) to her DA assignment.
At termination, every student in \(B^*\supseteq \BB(\jbc(P))\) strictly prefers \(\mu^*\) to \(\da(P)\).
The refinement step permutes schools only among \(B^*\) via cycles in which each participant improves over her current assignment, hence strictly improves over \(\da(P)\).
Students outside \(B^*\) retain their DA assignments throughout.
Thus \(\mu^+\) Pareto-dominates \(\da(P)\).

\noindent [Part 2] Let \(\Pi^*\) be the cycle packing induced by the terminal expansion matching \(M^*\).
Every non-loop edge \(i\to j\) used by \(M^*\) is admissible at the terminal step, so \(l(i\to j)\subseteq B^*=\BB(\mu^*)\).
Hence \(l(\Pi^*)\subseteq \BB(\mu^*)\).
By Lemma~\ref{prop:cc_packing}, \(\mu^*\) is justifiable.
Since the refinement permutes assignments only among \(B^*\) and each refinement-cycle participant strictly improves over her current school, \(\BB(\mu^+)=B^*\).
It suffices to show that no improvable non-beneficiary suffers a priority violation under $\mu^+$.
Suppose for contradiction that $h\in\IR(P)\setminus B^*$ prefers school $s=\mu^+_i$ to $\mu^+_h=\da_h(P)$ and $\rk_s(h)<\rk_s(i)$.
If $\mu^*_i=s$, the same violation exists under $\mu^*$, contradicting its justifiability.
Otherwise $i$ received $s$ during the refinement phase: specifically, $i$ took the school of some student $j$ along a cycle of admissible edges.
Because $h$ is improvable, $h\neq j$ (since $j\in B^*$), $h$ prefers $s$ to $\da_h(P)$, and $\rk_s(h)<\rk_s(i)$, admissibility of the edge $i\to j$ requires $h\in B^*$, a contradiction.

\noindent \josue{[Part 3] Let $\delta=\da(P)$, $\mu=\mu^+$, $B=\BB(\mu)=B^*$, and $R=\mathcal I(P)\setminus B$.
	By Part~2 and Lemma~\ref{prop:cc_packing}, the cycle packing representing $\mu$ relative to $\delta$ uses only edges with labels contained in $B$.
	Thus $\mu$ remains feasible in the terminal expansion graph $H(B)$, with exactly $|B|$ non-loop edges.}

\josue{Suppose that a justifiable matching $\nu$ Pareto-dominates $\mu$, and let
	\[
	D=\{i\in I:\nu_i\succ_i\mu_i\}.
	\]
	Strict preferences imply that $D$ consists precisely of the students whose assignments change, and $\varnothing\ne D\subseteq\mathcal I(P)$.
	Since $\mu$ is non-wasteful, no improving student can enter a school with a vacancy under $\mu$; a full school cannot gain students either.
	Moreover, $\mu$ is individually rational, so no student can improve by becoming unassigned.
	Consequently, every school has the same occupancy under $\mu$ and $\nu$.
	For each school $s$, therefore, choose a bijection
	\[
	a_s:\{j\in D:\mu_j=s\}\longrightarrow\{i\in D:\nu_i=s\},
	\]
	and write $a(j)=a_{\mu_j}(j)$.
	Then $a(j)\in D$ and $\nu_{a(j)}=\mu_j$.
	This construction accommodates arbitrary quotas.}

\josue{For each $j\in D$, let $c(j)$ be the highest-priority student at $\mu_j$ among
	\[
	\{a(j)\}\ \cup\ \{h\in R:\mu_j\succ_h\delta_h\}.
	\]
	We claim that $c(j)\in D$.
	This is immediate if $c(j)=a(j)$.
	Otherwise, $c(j)\in R$ and $c(j)\triangleright_{\mu_j}a(j)$.
	If $c(j)\notin D$, then
	\[
	\nu_{c(j)}=\mu_{c(j)}=\delta_{c(j)}.
	\]
	Since $c(j)$ prefers $\mu_j$ to this assignment, assigning $\mu_j$ to $a(j)$ under $\nu$ would violate the priority of an improvable non-beneficiary, contradicting justifiability.
	Hence $c(j)\in D$.}

\josue{In either case, $c(j)$ strictly prefers $\mu_j$ to $\mu_{c(j)}$.
	Moreover, by its selection, no student in $R$ who prefers $\mu_j$ to her DA assignment has higher priority than $c(j)$ at that school.
	The edges $c(j)\to j$ therefore form a graph on $D$ in which every vertex has an incoming edge and every edge is a strictly improving trade that respects all priorities in $R$.
	This finite graph contains a directed cycle $C$.}

\josue{If $C\subseteq B$, every edge of $C$ is admissible for the refinement phase, contradicting its termination.
	Otherwise, execute $C$ starting from $\mu$.
	All existing beneficiaries remain strictly better off than under DA, and at least one student in $R$ becomes a beneficiary.
	Every new DA-relative trading edge has its label contained in $B$, while the unchanged edges remain admissible because $\mu$ is feasible in $H(B)$.
	The resulting matching is therefore feasible in $H(B)$, preserves every beneficiary in $B$, and contains more than $|B|$ non-loop edges.
	This contradicts the terminal expansion step's maximization.
	Thus no justifiable matching Pareto-dominates $\mu^+$.}

\noindent [Part 4] The expansion phase initializes \(B^1=\BB(\jbc(P))\).
At every expansion step, the optimization forbids every student in \(B^t\) from using her self-loop, so \(B^t\subseteq B^{t+1}\).
Therefore every student in \(B^1\) remains a beneficiary under \(\mu^*\).
The refinement step implements trades only among students in \(B^*\), and every participant in a refinement cycle strictly improves over her current assignment.
Hence students in \(\BB(\jbc(P))\subseteq B^*\) remain beneficiaries under \(\mu^+\), so \(\BB(\jbc(P))\subseteq \BB(\mu^+)\).

\noindent [Part 5] DA runs in \(O(|I|\cdot |S|)\) time, since each student applies to each school at most once, so there are at most \(|I|\cdot |S|\) proposals in total, and each proposal is processed once.

The labelled envy digraph has at most \(|I|(|I|-1)\) edges.
For each edge \(i \to j\), its label is obtained by scanning the set of improvable students who prefer \(\da_j(P)\) to their DA assignment and checking which of them have higher priority than \(i\) at \(\da_j(P)\).
Thus all labels can be computed in polynomial time.

JBC constructs the auxiliary digraph on \(S^*(P)\) by assigning to each school \(s \in S^*(P)\) the unique outgoing edge to \(\da_{i_s^*}(P)\).
Since this digraph has at most \(|S|\) nodes and out-degree one at every node, its directed cycles can be identified in \(O(|S|)\) time.

At each iteration, the expansion phase computes a perfect matching of the augmented bipartite graph \(H^t\) that maximizes the number of non-loop edges subject to the self-loop restrictions for current beneficiaries.
This is a standard assignment problem, equivalently a minimum-cost perfect matching problem in a bipartite graph, and is solvable in polynomial time.
The expansion phase has at most \(|I|\) successful iterations, since \(B^t\) strictly increases whenever the algorithm does not stop.
Therefore the expansion phase runs in polynomial time.

For the refinement phase, fix an iteration \(r\).
The restricted envy digraph on \(B^*\) has at most \(|B^*|(|B^*|-1)\) edges.
Admissibility of each edge \(i\to j\) is checked by scanning the students in \(\mathcal I(P)\setminus\{j\}\) and verifying preference and priority at school \(\mu^r_{\mathrm{ref},j}\).
Hence the admissible subgraph can be computed in polynomial time.

A maximal collection of vertex-disjoint directed cycles in the admissible subgraph can also be found in polynomial time: repeatedly find any directed cycle in the current graph by depth-first search, record it, delete its vertices, and continue until no directed cycle remains.
Since each search and deletion step is polynomial and at most \(|B^*|\) vertices are deleted overall, this procedure is polynomial.

At each refinement iteration, every student who moves along a selected cycle strictly prefers her new assignment to her current one.
Therefore the sum of the ranks of the current assignments of students in \(B^*\) strictly decreases at every successful refinement iteration.
Since each such rank lies between \(1\) and \(|S|\), this sum can decrease at most \(|B^*|(|S|-1)\) times.
Hence the refinement phase performs at most \(|B^*|(|S|-1)\) successful iterations, and each iteration is polynomial-time.

Therefore SJBC+ runs in time polynomial in \(|I|\) and \(|S|\).
\end{proof}

\clearpage

\section{Proof of Theorem \ref{prop:cc_vs_eada}, Part 2}
\label{app:consenttree}
We show that the allocation corresponding to the cycle packing
\[
\{(i_1 \rightarrow i_2 \rightarrow i_1),\; (i_3 \rightarrow i_6 \rightarrow i_4 \rightarrow i_5 \rightarrow i_3)\}
\]
cannot be obtained by EADA under any consent set, to finalize the proof of Theorem \ref{prop:cc_vs_eada}.
Figure~\ref{fig:eada-tree} enumerates all possible EADA outcomes, starting from the last interrupter, who is asked if she consents to waive her priorities or not.

\begin{figure}[h!]
\centering
\input{tree}
\caption{EADA outcomes for all consent sets in Example~\ref{ex:running}.}
\label{fig:eada-tree}
\end{figure}
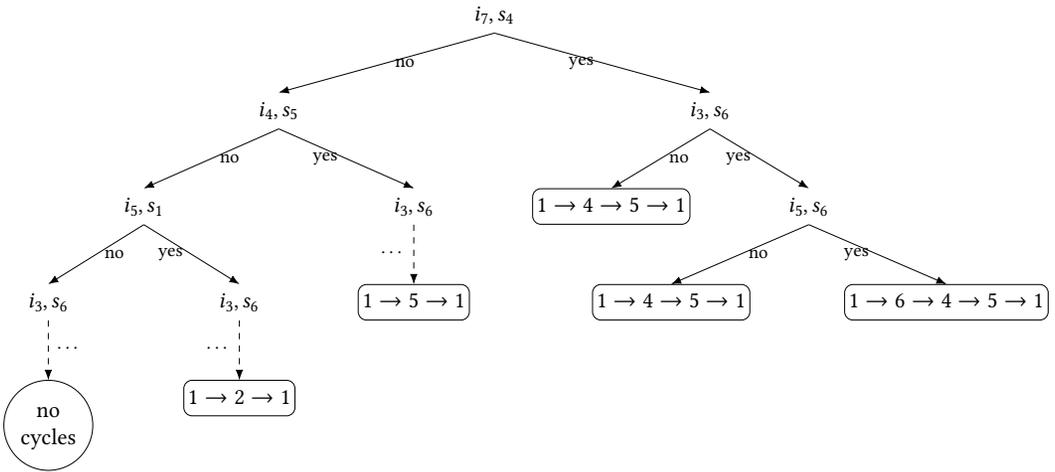

\clearpage

\newpage

\section{Proof of Proposition \ref{prop:top-manipulation}}
\label{app:top}

\begin{proof}

Let $M$ denote either JBC or a fixed selection of SJBC$+$.
Fix a student $i$ and a school $s\succcurlyeq_i M_i(P)$, and let $P'$ denote the problem obtained when $i$ reports $\succ_i^s$.
Since both mechanisms weakly Pareto-dominate DA,
\[
s\succcurlyeq_i M_i(P)\succcurlyeq_i \da_i(P).
\]
We distinguish two cases.

\medskip
\noindent\emph{Case 1: $s\succ_i\da_i(P)$.}
Let $\mu=\da(P)$ and $\mu'=\da(P')$.
Moving $s$ to the top changes only the ordering of schools that $i$ ranks above $\mu_i$, and leaves the set of such schools unchanged.

We first show that $\mu'=\mu$.
The matching $\mu$ remains stable under $P'$: preferences of all students other than $i$ are unchanged, while the set of schools that $i$ prefers to $\mu_i$ is unchanged.
Since $\mu'$ is the student-optimal stable matching under $P'$, we have $\mu'_i\succcurlyeq_i^s\mu_i$.
Strategy-proofness of DA also implies
\[
\mu_i\succcurlyeq_i\mu'_i.
\]
Since the set of schools above $\mu_i$ is the same under $\succ_i$ and $\succ_i^s$, these two inequalities imply $\mu'_i=\mu_i$.

It follows that $\mu'$ is also stable under $P$, since the only potentially different blocking pairs are those involving $i$, and her upper contour set at $\mu_i$ is unchanged.
Student-optimality of $\mu$ under $P$ therefore implies that every student weakly prefers $\mu$ to $\mu'$.
Together with student-optimality of $\mu'$ under $P'$, and the fact that all other students' preferences are unchanged, this gives $\mu'=\mu$.

For JBC, the same argument implies that the DA envy graph is unchanged.
Hence the set of improvable students and all labels in the DA envy graph are unchanged.
Moreover, during DA a student is rejected by exactly the schools she ranks above her final DA assignment.
For student $i$ this set is unchanged, and nothing changes for any other student.
Thus $S^*(P)$, the cutoff students, the sets $A(s)$, and the just-below-cutoff students $i_s^*$ are all unchanged.
The JBC graph and its cycles are therefore identical, so
\[
\jbc(P')=\jbc(P).
\]

For SJBC$+$, the expansion phase is also unchanged: it starts from the same JBC outcome and uses the same DA envy graph and the same labels.
Under a fixed selection rule, it therefore generates the same sequence of matchings through the expansion phase.

Consider now the refinement phase.
As long as the two executions coincide, let $\mu^t$ denote their current matching.
Refinement only executes Pareto-improving cycles, so in the truthful execution
\[
M_i(P)\succcurlyeq_i \mu_i^t.
\]
Since $s\succcurlyeq_i M_i(P)$, whenever $\mu_i^t\neq s$ we have $s\succ_i\mu_i^t$.
Moving $s$ to the top therefore leaves unchanged both the set of schools that $i$ prefers to $\mu_i^t$ and the set she prefers to $\da_i(P)$.
Hence the envy graph at $\mu^t$ and all admissible edges are the same under $P$ and $P'$.
A fixed selection rule therefore chooses the same cycle.
Thus the two executions coincide until the first time, if any, that student $i$ receives $s$.

If $i$ never receives $s$, the two executions coincide throughout and her assignment is unchanged.
If she does receive $s$, then she keeps $s$ under $P'$, since $s$ is her reported first choice.
She also cannot subsequently leave $s$ in the truthful execution: every subsequent move would have to strictly improve her assignment, which would imply $M_i(P)\succ_i s$, contradicting $s\succcurlyeq_i M_i(P)$.
Hence in this case $M_i(P)=M_i(P')=s$.
Thus SJBC$+$ is top-manipulation-proof.

\medskip
\noindent\emph{Case 2: $s=\da_i(P)$.}
Since
\[
s\succcurlyeq_i M_i(P)\succcurlyeq_i\da_i(P)=s,
\]
we have $M_i(P)=s$.
After $i$ moves $s$ to the top of her list, the original DA matching $\da(P)$ remains stable under $P'$: all other students' preferences are unchanged, while $i$ now receives her first choice at $\da(P)$.
Therefore student-optimality of DA under $P'$ implies
\[
\da_i(P')=s.
\]
Finally, both JBC and SJBC$+$ weakly Pareto-dominate DA at the reported preference profile, so
\[
M_i(P')\succcurlyeq_i^s \da_i(P')=s.
\]
Since $s$ is first under $\succ_i^s$, this implies $M_i(P')=s=M_i(P)$.

Thus, in both cases,
\[
M_i(P)\succcurlyeq_i M_i(P'),
\]
and JBC and any fixed selection of SJBC$+$ are top-manipulation-proof.	
\end{proof}

\clearpage
\setcounter{page}{1}
\pagenumbering{arabic}
\begin{center}
	{\bfseries Online Appendix}
\end{center}

\setcounter{section}{0}
\setcounter{subsection}{0}
\setcounter{subsubsection}{0}

\renewcommand{\appendixname}{Online Appendix}
\renewcommand{\thesection}{\arabic{section}}

%%%%%%%

\section{DA and EADA Descriptions}
\label{app:da_eada}

This appendix describes the Deferred Acceptance (DA) algorithm and the Efficiency-Adjusted Deferred Acceptance (EADA) algorithm.

\subsection{Deferred Acceptance (DA) Algorithm}
\label{app:da}

The student-proposing Deferred Acceptance algorithm proceeds as follows:
\begin{enumerate}
	\item \textbf{Round 1:} Each student applies to her most preferred school.
	Each school tentatively accepts the highest-priority applicants up to its capacity and rejects the rest.
	\item \textbf{Round $k \ge 2$:} Each rejected student applies to her next most preferred school.
	Each school considers its current tentative acceptances together with new applicants, tentatively accepts the highest-priority students up to its capacity, and rejects the rest.
	\item The algorithm terminates when no student is rejected.
	Tentative assignments become final.
\end{enumerate}

The outcome is the student-optimal stable matching $\da(P)$.

\subsection{Efficiency-Adjusted Deferred Acceptance (EADA) Algorithm}
\label{app:eada}

The EADA algorithm improves upon the DA outcome by allowing students to consent to waive priorities that do not affect their own assignments.
Let $W \subseteq I$ be the set of consenting students.

\begin{enumerate}
	\item \textbf{Round 0:} Run the DA algorithm with original preferences.
	Let $\mu^0$ be the outcome.
	\item \textbf{Round $r \ge 1$:}
	\begin{itemize}
		\item Examine the DA execution from Round $r-1$.
		\item Find the {last step} where a consenting student $i \in W$ is rejected from a school $s$ for which $i$ is an {interrupter} (i.e., $i$ was tentatively placed at $s$ earlier, and at least one other student was rejected from $s$ while $i$ was tentatively placed there).
		\item Identify all such interrupting pairs $(i, s)$ with $i \in W$.
		\item If none exist, stop and output the current matching.
		\item Otherwise, for each such $(i, s)$, remove school $s$ from $i$'s preference list (keeping the relative order of other schools unchanged).
		\item Rerun DA with the updated preferences.
	\end{itemize}
	\item The algorithm terminates in finite time.
	The final matching is $\eada(P, W)$.
\end{enumerate}

The EADA mechanism satisfies the following properties:
\begin{itemize}
	\item $\eada(P, W)$ weakly Pareto dominates $\da(P)$ for any $W$.
	\item No non-consenting student's priority is ever violated.
	\item A consenting student is never harmed by consenting: $\eada_i(P, W \cup \{i\}) \succcurlyeq_i \eada_i(P, W)$ for all $i$.
	\item If all students consent ($W = I$), the outcome is Pareto-efficient.
	
\end{itemize}

\newpage 
%%%%%%%%%%%%%%%%%%%%%%%%%%%%%%%%%%%%%%%%%%%%%%%%%%%%%%
\section{Detailed Execution of EADA in Example~\ref{ex:running}}
\label{app:eada-exec}
%%%%%%%%%%%%%%%%%%%%%%%%%%%%%%%%%%%%%%%%%%%%%%%%%%%%%%

This appendix reports a detailed execution of Kesten's EADA procedure on our running Example~\ref{ex:running}.
We first compute $\da(P)$ and then illustrate the successive EADA iterations under (i) full consent and (ii) a restricted consent set.

\subsection{Computing $\da(P)$ (Iteration 0)}
The DA proposals evolve as follows.

\begin{center}
	\begin{tabular}{lccccccc}
		\bf Iteration 0&&&&&&&\\
		round& $s_1$ & $s_2$ & $s_3$ & $s_4$ & $s_5$ & $s_6$ & $s_7$\\
		\hline
		r1  & 2 &     & 5 & \underline{6},7 & 4 & \underline{1},3 & \\
		r2  & 2 &     & 5 & \underline{1},7 & 4 & \bf \underline{3},6 & \\
		r3  & 2 & 1   & \bf 3,\underline{5} & 7 & 4 & 6 & \\
		r4  & 2 & 1   & 3 & 7 & 4 & \bf \underline{5},6 & \\
		r5  & 2 & 1   & 3 & \underline{5},7 & 4 & 6 & \\
		r6  & \underline{2},5 & 1 & 3 & 7 & 4 & 6 & \\
		r7  & 5 & \bf \underline{1},2 & 3 & 7 & 4 & 6 & \\
		r8  & 5 & 2 & \bf \underline{1},3 & 7 & 4 & 6 & \\
		r9  & 5 & 2 & 3 & 7 & \underline{1},4 & 6 & \\
		r10 & \bf 1,\underline{5} & 2 & 3 & 7 & 4 & 6 & \\
		r11 & 1 & 2 & 3 & 7 & \bf \underline{4},5 & 6 & \\
		r12 & 1 & 2 & 3 & \bf 4,\underline{7} & 5 & 6 & \\
		r13 & 1 & 2 & 3 & 4 & 5 & 6 & 7\\
		\hline
	\end{tabular}
\end{center}
\medskip

\subsection{EADA under full consent}
Fix a consent set $W$ that contains every potential interrupter (equivalently, assume that whenever EADA queries an interrupter, the interrupter consents).
Kesten's procedure identifies the {last interrupter in $W$} at the DA outcome, asks that student to waive priority at the relevant school, removes that school from the student's preference list, and reruns DA on the modified problem.

In Iteration~0 the last interrupter is $i_7$ at $s_4$.
Since $i_7\in W$, we delete $s_4$ from $i_7$'s preference list and rerun DA.

\begin{center}
	\begin{tabular}{lccccccc}
		\bf Iteration 1&&&&&&&\\
		round& $s_1$ & $s_2$ & $s_3$ & $s_4$ & $s_5$ & $s_6$ & $s_7$\\
		\hline
		r1 & 2 & & 5 & 6 & 4 & \underline{1},3 & 7\\
		r2 & 2 & & 5 & 1,\underline{6} & 4 & 3 & 7\\
		r3 & 2 & & 5 & 1 & 4 & \underline{3},6 & 7\\
		r4 & 2 & & 3,\underline{5} & 1 & 4 & 6 & 7\\
		r5 & 2 & & 3 & 1 & 4 & \underline{5},6 & 7\\
		r6 & 2 & & 3 & 1,\underline{5} & 4 & 6 & 7\\
		r7 & \underline{2},5 & & 3 & 1 & 4 & 6 & 7\\
		r8 & 5 & 2 & 3 & 1 & 4 & 6 & 7\\
		\hline
	\end{tabular}
\end{center}
\medskip

At this modified DA outcome the last interrupter is $i_3$ at $s_6$.
We therefore delete $s_6$ from $i_3$'s preference list and rerun DA.

\begin{center}
	\begin{tabular}{lccccccc}
		\bf Iteration 2&&&&&&&\\
		round& $s_1$ & $s_2$ & $s_3$ & $s_4$ & $s_5$ & $s_6$ & $s_7$\\
		\hline
		r1 & 2 & & 3,\underline{5} & 6 & 4 & 1 & 7\\
		r2 & 2 & & 3 & 6 & 4 & \underline{1},5 & 7\\
		r3 & 2 & & 3 & 1,\underline{6} & 4 & 5 & 7\\
		r4 & 2 & & 3 & 1 & 4 & \underline{5},6 & 7\\
		r5 & 2 & & 3 & 1,\underline{5} & 4 & 6 & 7\\
		r6 & \underline{2},5 & & 3 & 1 & 4 & 6 & 7\\
		r7 & 5 & 2 & 3 & 1 & 4 & 6 & 7\\
		\hline
	\end{tabular}
\end{center}
\medskip

Now the last interrupter is $i_5$ at $s_6$.
Deleting $s_6$ from $i_5$'s preference list and rerunning DA yields:

\begin{center}
	\begin{tabular}{lccccccc}
		\bf Iteration 3&&&&&&&\\
		round& $s_1$ & $s_2$ & $s_3$ & $s_4$ & $s_5$ & $s_6$ & $s_7$\\
		\hline
		r1 & 2 & & 3,\underline{5} & 6 & 4 & 1 & 7\\
		r2 & 2 & & 3 & \underline{5},6 & 4 & 1 & 7\\
		r3 & \underline{2},5 & & 3 & 6 & 4 & 1 & 7\\
		r4 & 5 & 2 & 3 & 6 & 4 & 1 & 7\\
		\hline
	\end{tabular}
\end{center}
\medskip

This outcome implements the cycle
\[
(i_1 \rightarrow i_6 \rightarrow i_4 \rightarrow i_5 \rightarrow i_1),
\]
which yields a Pareto-efficient allocation.

\subsection{EADA under partial consent}
We now illustrate the failure of EADA to achieve Pareto-efficiency under a restricted consent set that we use in the proof of Theorem \ref{thm:bigimpo}.
Let $W=\{i_1,i_5,i_7\}$.

EADA starts from $\da(P)$ (Iteration~0) and identifies the last interrupter at that outcome, which is again $i_7$ at $s_4$.
Since $i_7\in W$, we delete $s_4$ from $i_7$'s preference list and rerun DA, obtaining Iteration~1 above.

At the resulting DA outcome, the (unique) last interrupter is $i_3$ at $s_6$, but $i_3\notin W$.
By the definition of EADA with an exogenous consent set, the procedure stops at this point.
The resulting allocation corresponds to the cycle
\[
(i_1 \rightarrow i_4 \rightarrow i_5 \rightarrow i_1),
\]
highlighted in bold in the table below.

\begin{center}
	\begin{tabular}{ccccccc}
		\hline
		$i_1$ & $i_2$ & $i_3$ & $i_4$ & $i_5$ & $i_6$ & $i_7$\\
		\hline
		$s_6$ & $s_1$ & $s_6$ & $\bm{s_5}$ & $s_3$ & $s_4$ & $s_4$\\
		$\bm{s_4}$ & $\bm{s_2}$ & $\bm{s_3}$ & $s_4$ & $s_6$ & $\pmb{s_6}$ & $\pmb{s_7}$\\
		$s_2$ & $\cdot$ & $\cdot$ & $\cdot$ & $s_4$ & $\cdot$ & $\cdot$\\
		$s_3$ & $\cdot$ & $\cdot$ & $\cdot$ & $\bm{s_1}$ & $\cdot$ & $\cdot$\\
		$s_5$ & $\cdot$ & $\cdot$ & $\cdot$ & $s_5$ & $\cdot$ & $\cdot$\\
		$s_1$ & $\cdot$ & $\cdot$ & $\cdot$ & $\cdot$ & $\cdot$ & $\cdot$\\
		\hline
	\end{tabular}
\end{center}
\medskip

This allocation is not Pareto-efficient: for instance, $i_1$ (assigned $s_4$) and $i_6$ (assigned $s_6$) can exchange assignments and both strictly gain.

Nevertheless, under the same consent set $W=\{i_1,i_5,i_7\}$ there exists a Pareto-efficient matching that weakly Pareto-dominates $\da(P)$ and respects the priorities of every student in $I\setminus W$, obtained by implementing the two disjoint cycles
\[
(i_1 \rightarrow i_2 \rightarrow i_1)
\qquad\text{and}\qquad
(i_3 \rightarrow i_6 \rightarrow i_4 \rightarrow i_5 \rightarrow i_3).
\]
Hence $\mathcal{M}(P,W)\neq\emptyset$ while $\eada(P,W)$ is not Pareto-efficient, establishing the claim used in Proposition~\ref{prop:bigproblem}.

\newpage

\newpage
%%%%%%%%%%%%%%%%%%
\section{Why SJBC Requires the ``\,+\,'' Step}
\label{app:why-plus}
%%%%%%%%%%%%%%%%%%%%%%%%%%%

This appendix explains why the final step in SJBC$+$ is necessary.
In particular, the naive procedure that iterates JBC need not return a Pareto-efficient justifiable matching, even when such a matching exists.
The issue is that an early locally-justifiable trade can change the residual instance in a way that prevents reaching a Pareto-efficient outcome within the class of justifiable improvements over the DA benchmark.
SJBC$+$ addresses this by adding a final Pareto-improvement step that preserves the set of beneficiaries.

Consider the following problem with five students and five schools, each of unit capacity.
Preferences are listed in the left block and priorities in the right block; the DA assignment is indicated in bold.

\begin{example}[Preferences and priorities]\label{ex:whyplus}
	\centering
	\begin{tabular}{ccccc|ccccc}
		\hline
		$i_1$ & $i_2$ & $i_3$ & $i_4$ & $i_5$ & $s_1$ & $s_2$ & $s_3$ & $s_4$ & $s_5$ \\
		\hline
		$s_5$ & $s_4$ & $s_2$ & $s_5$ & $s_5$ & $i_1$ & $i_5$ & $\bm{i_1}$ & $\bm{i_5}$ & $\bm{i_3}$ \\
		$\pmb{s_3}$ & $s_5$ & $s_4$ & $\bm{s_2}$ & $s_1$ & $\bm{i_2}$ & $\bm{i_4}$ & $\cdot$ & $i_4$ & $i_1$ \\
		$\cdot$ & $s_2$ & $s_1$ & $\cdot$ & $\bm{s_4}$ & $i_5$ & $i_3$ & $\cdot$ & $i_1$ & $i_4$ \\
		$\cdot$ & $\bm{s_1}$ & $\bm{s_5}$ & $\cdot$ & $\cdot$ & $i_4$ & $i_2$ & $\cdot$ & $i_3$ & $i_5$ \\
		$\cdot$ & $\cdot$ & $\cdot$ & $\cdot$ & $\cdot$ & $i_3$ & $i_1$ & $\cdot$ & $i_2$ & $i_2$ \\
	\end{tabular}
	
\end{example}

Student $i_1$ is unimprovable at the DA outcome (no other student envies $i_1$), so Sequential JBC deletes $i_1$ together with the school assigned to $i_1$ under DA (here, $s_3$), and proceeds with the reduced instance on $\{i_2,i_3,i_4,i_5\}$.
The resulting envy digraph is depicted in Figure~\ref{fig:envy-example3}.

\begin{figure}[h!]
	\centering
	\begin{tikzpicture}[scale=0.65, ->, node distance=2.5cm,
		every node/.style={draw, circle, minimum size=0.9cm}, thick]
		\def\radius{2.5}
		\def\n{4}
		
		\foreach \i/\name in {1/i_2, 2/i_3, 3/i_4, 4/i_5} {
			\node (\name) at ({90 + 360/\n * (\i-1)}:\radius) {$\name$};
		}
		
		\draw[bend left=15] (i_2) to node[blue, pos=0.5, draw=none]{$\bm{i_4, i_5}$} (i_3);
		\draw (i_2) to node[blue,pos=0.8, draw=none]{$\bm{i_3}$} (i_4);
		\draw[bend left] (i_2) to node[blue, midway, above, draw=none, yshift=-8pt]{$\bm{i_3}$} (i_5);
		
		\draw[bend left] (i_3) to node[blue, pos=0.5, above, draw=none, yshift=-8pt,xshift=-4pt]{$\bm{i_5}$} (i_2);
		\draw[ultra thick] (i_3) to node[blue, pos=0.5, draw=none]{$\bm \emptyset$} (i_4);
		\draw[bend left=15] (i_3) to node[blue, draw=none, xshift=15pt]{$\bm{\emptyset}$} (i_5);
		
		\draw[bend left=30, ultra thick] (i_4) to node[blue, pos=0.5, draw=none]{$\bm \emptyset$} (i_3);
		
		\draw (i_5) to node[blue, pos=0.5, draw=none]{$\bm \emptyset$} (i_2);
		\draw[bend left=15] (i_5) to node[blue, pos=0.2,  draw=none]{$\bm{i_4}$} (i_3);
	\end{tikzpicture}
	\caption{Envy digraph for Example \ref{ex:whyplus} after deleting $i_1$ and $s_3$.}
	\label{fig:envy-example3}
\end{figure}
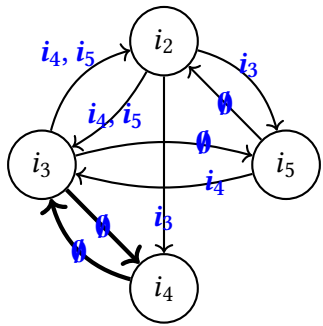

In this reduced digraph, JBC selects the $\emptyset$-labelled 2-cycle between $i_3$ and $i_4$, recommending that $i_3$ and $i_4$ exchange schools.
This yields the intermediate matching
\[
\mu' = (i_1\rightarrow s_3,\; i_2\rightarrow s_1,\; i_3\rightarrow s_2,\; i_4\rightarrow s_5,\; i_5\rightarrow s_4).
\]
At this point, Sequential JBC continues by selecting a maximum packing of admissible cycles in the {original} envy digraph that is consistent with the set of beneficiaries already created (here, $i_3$ and $i_4$).
One feasible packing is
\[
\bigl(i_4 \stackrel{\emptyset}{\rightarrow} i_3 \stackrel{\emptyset}{\rightarrow} i_4\bigr), \qquad
\bigl(i_2 \stackrel{i_3}{\rightarrow} i_5 \stackrel{\emptyset}{\rightarrow} i_2\bigr),
\]
which remains justifiable.
The resulting allocation is Pareto-efficient; for instance, it coincides with the outcome of serial dictatorship under the order $i_2,i_3,i_4,i_5,i_1$.

However, Sequential JBC can instead select the (also admissible and justifiable) cycle
\[
\bigl(i_2 \stackrel{i_3}{\rightarrow} i_4 \stackrel{\emptyset}{\rightarrow} i_3
\stackrel{\emptyset}{\rightarrow} i_5 \stackrel{\emptyset}{\rightarrow} i_2\bigr),
\]
which induces the corresponding allocation shown below:
\begin{center}
	\begin{tabular}{ccccc}
		$i_1$ & $i_2$ & $i_3$ & $i_4$ & $i_5$ \\
		\hline
		$s_5$ & ${s_4}$ & ${s_2}$ & $\bm{s_5}$ & $s_5$ \\
		$\pmb{s_3}$ & $s_5$ & $\bm{s_4}$ & ${s_2}$ & $\bm{s_1}$ \\
		$\cdot$ & $\bm{s_2}$ & $s_1$ & $\cdot$ & ${s_4}$ \\
		$\cdot$ & ${s_1}$ & ${s_5}$ & $\cdot$ & $\cdot$ \\
		$\cdot$ & $\cdot$ & $\cdot$ & $\cdot$ & $\cdot$ \\
	\end{tabular}
\end{center}
This outcome is {not} Pareto-efficient: students $i_2$ and $i_3$ can exchange assignments and both strictly gain.

The failure arises from {selection} among admissible cycle packings: Sequential JBC can terminate at a justifiable matching that is Pareto-dominated by another justifiable matching (and, crucially, without changing the set of beneficiaries).
This motivates the final ``\,+\,'' step in SJBC$+$: after Sequential JBC determines the set of beneficiaries, SJBC$+$ performs an additional Pareto-improvement stage (restricted to justifiable trades that preserve that beneficiary set), ensuring constrained efficiency within the class of justifiable improvements over DA.

\newpage

%%%%%%%%%%%%%%%%%%%%%%%%%%%%%%%%%%%%%%%%%%%%%%%%%%%%%%%%%%%%%%%%%%%%%%%%
\section{Justifiability Does Not Imply Preserving JBC Beneficiaries}
\label{app:justifiable_not_core}
%%%%%%%%%%%%%%%%%%%%%%%%%%%%%%%%%%%%%%%%%%%%%%%%%%%%%%%%%%%%%%%%%%%%%%%%
This appendix provides an instance in which there exists a justifiable matching that Pareto-improves on $\da(P)$, but whose beneficiary set is not nested with that of $\jbc(P)$.
In particular, there is a justifiable improvement that benefits a student who does not benefit under $\jbc(P)$, while $\jbc(P)$ benefits students who do not benefit under that justifiable improvement.
Hence, justifiability alone does not force a reform to preserve the JBC beneficiary set.

Consider the following one-to-one instance with $I=\{i_1,\dots,i_6\}$ and $S=\{s_1,\dots,s_6\}$.
Columns on the left list students' preferences (top to bottom), and columns on the right list schools' priorities (top to bottom).
Bold entries on the left indicate DA assignments $\da(P)$, and bold entries on the right indicate the students assigned to each school under DA.

\begin{example}[Preferences and Priorities]
	\label{ex:notcorebenefi}
	\centering
	\begin{tabular}{cccccc|cccccc}
		\hline
		$i_1$ & $i_2$ & $i_3$ & $i_4$ & $i_5$ & $i_6$			& $s_1$ & $s_2$ & $s_3$ & $s_4$ & $s_5$ & $s_6$ \\
		\hline
		$s_4$ & $s_1$ & $s_5$ & $s_6$ & $s_4$ & $s_1$			& \bm{$i_5$} & \bm{$i_3$} & \bm{$i_4$} & \bm{$i_2$} & \bm{$i_1$} & \bm{$i_6$} \\
		$s_6$ & \bm{$s_4$} & $s_6$ & $s_2$ & \bm{$s_1$} & $s_2$	& $i_4$ & $i_2$ & $\cdot$ & $i_4$ & $i_3$ & $i_3$ \\
		\bm{$s_5$} & $s_6$ & $s_4$ & $s_4$ & $s_6$ & \bm{$s_6$}	& $i_3$ & $i_4$ & $\cdot$ & $i_5$ & $i_2$ & $i_1$ \\
		$s_3$ & $s_2$ & \bm{$s_2$} & $s_5$ & $s_5$ & $s_5$		& $i_2$ & $i_5$ & $\cdot$ & $i_6$ & $i_4$ & $i_2$ \\
		$s_1$ & $s_3$ & $s_1$ & $s_1$ & $s_2$ & $s_4$			& $i_6$ & $i_6$ & $\cdot$ & {$i_3$} & $i_6$ & $i_5$ \\
		$s_2$ & $s_5$ & $s_3$ & \bm{$s_3$} & $s_3$ & {$s_3$}	& $i_1$ & $i_1$ & $\cdot$ & $i_1$ & $i_5$ & $i_4$ \\
	\end{tabular}
\end{example}

The DA outcome is
\[
\mu^0=\da(P):\quad
(i_1\rightarrow s_5,\ i_2\rightarrow s_4,\ i_3\rightarrow s_2,\ i_4\rightarrow s_3,\ i_5\rightarrow s_1,\ i_6\rightarrow s_6).
\]

Applying JBC yields
\[
\jbc(P):\quad
(i_1\rightarrow s_5,\ i_2\rightarrow s_1,\ i_3\rightarrow s_6,\ i_4\rightarrow s_3,\ i_5\rightarrow s_4,\ i_6\rightarrow s_2),
\]

The corresponding beneficiary set is
\[
\BB(\jbc(P))=\{i_2,i_3,i_5,i_6\}.
\]

Consider the following DA-improving matching:
\[
\mu:\quad
(i_1\rightarrow s_6,\ i_2\rightarrow s_4,\ i_3\rightarrow s_5,\ i_4\rightarrow s_3,\ i_5\rightarrow s_1,\ i_6\rightarrow s_2).
\]

Its beneficiary set is
\[
\BB(\mu)=\{i_1,i_3,i_6\}.
\]

Therefore the beneficiary sets are not nested:
\[
\BB(\mu)\setminus \BB(\jbc(P))=\{i_1\}\neq\emptyset,
\qquad
\BB(\jbc(P))\setminus \BB(\mu)=\{i_2,i_5\}\neq\emptyset.
\]

In particular, $\mu$ benefits a student ($i_1$) who does not benefit under JBC, while JBC benefits students ($i_2,i_5$) who do not benefit under $\mu$.

This example shows that justifiability alone does not force an improvement that preserves the JBC beneficiary set.

\newpage

%%%%%%%%%%%%%%%%%%%%%%%%%%%%%%%%%%%%%%%%%%%%%%%%%%%%%%%%%%%%%%%%%%%%%%%%
\section{Essential Stability and Justifiability}
\label{app:essential}

In the main text we stated the existing result that EADA with full consent is essentially stable, which means that any claim starting from a priority violation is vacuous, i.e. ends up with the claimant losing her claim \citep{troyan2020essentially}.
This directly implies that essentially stable matchings need not be justifiable by Theorem \ref{prop:cc_vs_eada}.
Now we show that justifiable and Pareto-efficient matchings need not be essentially stable by constructing a non-vacuous reassignment chain starting from a priority violation in Example \ref{ex:running}.

Consider the unique Pareto-efficient and justifiable matching obtained with the following trading cycles:
\[
\{(i_1 \stackrel{\emptyset}{\rightarrow} i_2 \stackrel{i_5}{\rightarrow} i_1),(i_3 \stackrel{\emptyset}{\rightarrow} i_6 \stackrel{i_1}{\rightarrow} i_4 \stackrel{\emptyset}{\rightarrow} i_5 \stackrel{\emptyset}{\rightarrow} i_3)\}
\] 

Let us start with someone who suffers a priority violation making a claim.
Say $i_1$ at $s_4$.
This claim displaces $i_6$, who in turn claims $s_6$.
This last claim displaces $i_3$, who in turn claims $s_3$, displacing $i_5$.
In turn, $i_5$ claims $s_1$, displacing $i_2$, who now can go to the empty school $s_2$.
We conclude that the claim is not vacuous, which in turn implies that the unique justifiable and Pareto-efficient matching need not be essentially stable.

%%%%%%%%%%%%%%%%%%%%%%%%%%%%%%%
\newpage
%%%%%%%%%%%%%%%%%%%%%%%%%%%%%%%%%%%%%%%%%%%%%%%%%%%%%%%%%%%%
\section{Administrative Exercise Using the Chilean SAE}
\label{app:chile}
%%%%%%%%%%%%%%%%%%%%%%%%%%%%%%%%%%%%%%%%%%%%%%%%%%%%%%%%%%%%

This appendix describes the administrative exercise reported in Section~\ref{subsec:chile}.
We use publicly available microdata from Chile's Sistema de Admisión Escolar (SAE), the nationwide centralized assignment system for publicly funded schools.
Our purpose is not to reproduce every institutional feature of the official SAE procedure or to estimate the effect of a literal policy reform.
Rather, we use actual submitted preferences, observed applicant--school priority information, and school capacities to compare the mechanisms studied in the paper in a large real-world school choice environment.

%%%%%%%%%%%%%%%%%%%%%%%%%%%%%%%%%%%%%%%%%%%%%%%%%%%%%%%%%%%%
\paragraph{Institutional setting and data}
%%%%%%%%%%%%%%%%%%%%%%%%%%%%%%%%%%%%%%%%%%%%%%%%%%%%%%%%%%%%

Families participating in the SAE submit an ordered list of schools.
The assignment procedure is based on student-proposing deferred acceptance and uses legally defined priorities together with an official tie-breaking rule \citep{correa2019school}.
We use the 2025 regular application round for admission in the 2026 school year.

The public SAE files are released by the Chilean Ministry of Education.
We use the archive \texttt{SAE\_2025.rar}, which is catalogued, together with the other annual SAE releases, at
\url{https://datos.cochid.cl/dataset/mineduc_sistema_de_admision_escolar_sae}.

Our construction uses five files from the regular application round:
\begin{center}
	\begin{tabular}{ll}
		\toprule
		File & Information used\\
		\midrule
		\texttt{A1\_Oferta\ldots.csv}
		& School--year capacities.\\
		\texttt{B1\_Postulantes\ldots.csv}
		& Applicant characteristics.\\
		\texttt{C1\_Postulaciones\ldots.csv}
		& Submitted rankings and priorities.\\
		\texttt{D1\_Resultados\ldots.csv}
		& Published SAE assignments.\\
		\texttt{F1\_Relaciones\_\ldots.csv}
		& Linked sibling applications.\\
		\bottomrule
	\end{tabular}
\end{center}

The SAE covers applications from pre-kindergarten through the final year of secondary school.
We restrict the exercise to the two main entry cohorts: the first year of \emph{educación básica} and the first year of \emph{educación media}, corresponding to years 1 and 9 in the public data.
A seat is specific to both a school and the school year for which a child is applying.
Applicants in the two entry cohorts therefore never compete for the same seats, and we construct and run the assignment separately for $1^{\circ}$ básico and $1^{\circ}$ medio.

%%%%%%%%%%%%%%%%%%%%%%%%%%%%%%%%%%%%%%%%%%%%%%%%%%%%%%%%%%%%
\paragraph{Preferences and school priorities}
%%%%%%%%%%%%%%%%%%%%%%%%%%%%%%%%%%%%%%%%%%%%%%%%%%%%%%%%%%%%

We take submitted rank-order lists at face value.
We do not estimate latent preferences and do not complete truncated lists.
A school that does not appear on an applicant's submitted list is therefore treated as unacceptable.
When a school--year option is excluded from the analysis for the reasons described below, the positions of the remaining options are not renumbered.

The public application file reports the priority information separately for each applicant-school choice.
For ordinary seats, we order applicants lexicographically according to
\begin{enumerate}
	\item \texttt{prioridad\_matriculado},
	\item \texttt{prioridad\_hermano},
	\item \texttt{prioridad\_hijo\_funcionario},
	\item \texttt{prioridad\_exalumno},
	\item \texttt{loteria\_original},
	\item \texttt{mrun}.
\end{enumerate}
Thus, current enrolment is considered first, followed by whether the applicant has a sibling at the school, is the child of a school employee, or is a former student.
The variable \texttt{loteria\_original} is the official applicant--school tie-breaking value supplied in the public data.
The anonymized applicant identifier \texttt{mrun} is used only as a deterministic final tie-breaker if all preceding fields coincide.

The priority criteria are common features of the SAE, but the resulting priority ordering is school-specific.
For example, an applicant may have a sibling priority at one school and no sibling priority at another.
Likewise, the public tie-breaking value is recorded for the particular applicant--school choice.
Hence each retained school has its own strict ordering of the applicants who rank it.

%%%%%%%%%%%%%%%%%%%%%%%%%%%%%%%%%%%%%%%%%%%%%%%%%%%%%%%%%%%%
\paragraph{Constructing the analysis sample and DA benchmark}
%%%%%%%%%%%%%%%%%%%%%%%%%%%%%%%%%%%%%%%%%%%%%%%%%%%%%%%%%%%%

The official SAE contains several institutional features that are outside the model used in the paper, most importantly special or reserved seat categories and linked sibling applications.
We therefore restrict the mechanism comparison to occupied ordinary seats for which the public data permit us to construct the standard school-choice problem used in our theory.

For each of the two entry cohorts, we proceed as follows.
\begin{enumerate}
	\item We first reconstruct student-proposing DA using the submitted rankings, reported capacities, observed applicant--school priorities, and the reported number of seats reserved for priority students.
	
	\item We retain only school--year options with none of the principal special or reserved categories used by the SAE.
	In the public data this requires
	\begin{eqnarray*}
		\texttt{vacantes\_prioritarios}
		&=&\texttt{vacantes\_pie}=\\
		\texttt{vacantes\_alta\_exigencia\_t}
		&=&\texttt{vacantes\_alta\_exigencia\_r}=0.
	\end{eqnarray*}
	
	\item We retain applicants assigned by the reconstruction to one of these ordinary options, together with every ordinary school--year option holding at least one such applicant.
	
	\item Removing special-seat options and their applicants can make capacity appear available even though those seats belong to parts of the official assignment problem that we have excluded.
	We therefore set the capacity of each retained school--year option equal to the number of retained applicants occupying it in our reconstructed DA assignment.
	
	\item We restrict each retained applicant's submitted list to the retained options and rerun student-proposing DA.
	This matching is the DA benchmark against which beneficiaries, improvability, SJBC$+$, and EADA are defined.

\end{enumerate}

Across the two entry cohorts, this procedure retains 64,893 applicants and 5,496 occupied school--year options.
The $1^{\circ}$ básico sample contains 30,227 applicants and 3,760 occupied options, while the $1^{\circ}$ medio sample contains 34,666 applicants and 1,736 occupied options.
There are 10,954 improvable applicants in total, or 16.9\% of the retained sample: 1,853 in $1^{\circ}$ básico and 9,101 in $1^{\circ}$ medio.
The replication code also records the agreement between the final DA benchmark and the assignment inherited from the initial reconstruction as a validation check.
No applicant, school, or improvement component is selected on the basis of the performance of SJBC$+$ or EADA.

%%%%%%%%%%%%%%%%%%%%%%%%%%%%%%%%%%%%%%%%%%%%%%%%%%%%%%%%%%%%
\paragraph{Linked sibling applications}
%%%%%%%%%%%%%%%%%%%%%%%%%%%%%%%%%%%%%%%%%%%%%%%%%%%%%%%%%%%%

The SAE permits linked applicants to request joint treatment with siblings, something which we do not capture in the model.
For simplicity, we treat each child as an individual applicant using her own submitted ranking and observed applicant--school priorities.
The relationship file is used to identify applicants participating in linked applications, but we do not reproduce the joint assignment procedure.

For this reason, the exercise should be interpreted as applying the mechanisms to the preferences, priorities, and occupied ordinary seats observed in the Chilean data, rather than as a literal reconstruction of every institutional feature of the SAE.

%%%%%%%%%%%%%%%%%%%%%%%%%%%%%%%%%%%%%%%%%%%%%%%%%%%%%%%%%%%%
\paragraph{Mechanisms and consent}
%%%%%%%%%%%%%%%%%%%%%%%%%%%%%%%%%%%%%%%%%%%%%%%%%%%%%%%%%%%%

We report DA, SJBC$+$, EADA with full consent, and EADA with partial consent.

For full-consent EADA, every applicant consents.
For partial-consent EADA, we randomly select exactly $\lfloor 0.52N\rfloor$ applicants within each entry cohort, without replacement.
The 52\% calibration follows the consent treatment in \citet{cerrone2022school}.
We use 50 independent consent draws for each entry cohort.
The replication code fixes the random-number seed separately for each cohort, so the exercise is exactly reproducible.

We compute EADA using the simplified implementation of \citet{tang2014new}, which is outcome-equivalent to Kesten's interrupter-deletion procedure for a fixed consent set.

%%%%%%%%%%%%%%%%%%%%%%%%%%%%%%%%%%%%%%%%%%%%%%%%%%%%%%%%%%%%
\paragraph{Replication}
%%%%%%%%%%%%%%%%%%%%%%%%%%%%%%%%%%%%%%%%%%%%%%%%%%%%%%%%%%%%

The replication archive contains a single MATLAB file,\\
\texttt{replicate\_chile\_sae.m}, available at \url{www.josueortega.com}.

\newpage
%%%%%%%%%%%%%%%%%%%%%%%%%%%%%%%%%%%%%%%%%%%%%%%%%%%%%%%%%%%%
\section{JBC Is Not Upper-Manipulation-Proof}
\label{app:upper}
%%%%%%%%%%%%%%%%%%%%%%%%%%%%%%%%%%%%%%%%%%%%%%%%%%%%%%%%%%%%

\textcolor{black}{We provide here an example showing that JBC does not satisfy the stronger incentive property of upper-manipulation-proofness introduced by \citet{afacan2022improving}.
We first recall its definition.}

\textcolor{black}{For a preference order $\succ_i$ and a school $s$, let $L_s(\succ_i)\coloneqq\{s'\in S:s\succ_i s'\}$	denote the set of schools that student $i$ ranks below $s$.
A report $\succ_i'$ is an \emph{upper manipulation} of $\succ_i$ at school $s$ if
\begin{enumerate}
	\item $L_s(\succ_i)\subseteq L_s(\succ_i')$, and
	\item for every $s',s''\in L_s(\succ_i)$, $		s'\succ_i s''\quad\Longrightarrow\quad s'\succ_i' s''$.
\end{enumerate}
Thus, an upper manipulation may promote $s$ and reshuffle schools originally ranked above $s$, but it preserves both the position below $s$ and the relative order of all schools that were originally ranked below $s$.}

\begin{definition}
\textcolor{black}{A mechanism $M$ is \emph{upper-manipulation-proof} if, for every problem $P$, every student $i$, and every report $\succ_i'$ that is an upper manipulation of $\succ_i$ at $M_i(P)$,
	\[
	M_i(P)\succcurlyeq_i M_i(P'),
	\]
	where $P'$ is the problem obtained from $P$ by replacing $\succ_i$ with $\succ_i'$.}
\end{definition}

\textcolor{black}{The following example shows that JBC fails this property.}

\begin{example}[Preferences and priorities]
\label{ex:uppermanip}

\begin{center}
	\begin{tabular}{cccccc|cccccc}
		\hline
		$i_1$ & $i_2$ & $i_3$ & $i_4$ & $i_5$ & $i_6$		& $s_1$ & $s_2$ & $s_3$ & $s_4$ & $s_5$ & $s_6$ \\
		\hline
		$s_5$ 	& $s_5$ & $s_1$ & $s_2$ & $s_6$ & $s_1$		& $i_1$ & $i_3$ & $i_4$ & $i_2$ & $i_5$ & $i_6$ \\
		
		$s_1$ 	& $s_6$ & $s_4$ & $s_3$ & $s_5$ & $s_5$		& $i_3$ & $i_4$ & $\cdot$ & $i_3$ & $i_1$ & $i_5$ \\
		
		$\cdot$	 & $s_2$ & $s_2$ & $\cdot$ & $\cdot$ & $s_6$& $i_2$ & $i_2$ & $\cdot$ & $\cdot$ & $i_6$ & $i_2$ \\
		
		$\cdot$ & $s_1$ & $\cdot$ & $\cdot$ & $\cdot$ &$\cdot$& $i_6$ & $\cdot$ & $\cdot$ & $\cdot$ & $i_2$ & $\cdot$ \\
		
		$\cdot$ & $s_4$ & $\cdot$ & $\cdot$ & $\cdot$ & $\cdot$ 	& $\cdot$ & $\cdot$ & $\cdot$ & $\cdot$ & $\cdot$ & $\cdot$ \\
		
		\hline
	\end{tabular}
\end{center}
\end{example}

\textcolor{black}{Under truthful reporting,
\[
\da(P)=(i_1-s_1,i_2-s_4,i_3-s_2,i_4-s_3,i_5-s_5,i_6-s_6),
\]
while JBC generates
\[
\jbc(P)=(i_1-s_1,i_2-s_2,i_3-s_4,i_4-s_3,i_5-s_5,i_6-s_6).
\]
Thus student $i_6$ receives $s_6$, her third choice.
Her truthful preference order is
\[
s_1\succ_{i_6}s_5\succ_{i_6}s_6
\succ_{i_6}s_2\succ_{i_6}s_4\succ_{i_6}s_3.
\]
Suppose now that $i_6$ reports
\[
s_5\succ_{i_6}'s_6\succ_{i_6}'s_1
\succ_{i_6}'s_2\succ_{i_6}'s_4\succ_{i_6}'s_3.
\]
This is an upper manipulation at $s_6$.
Indeed,
\[
L_{s_6}(\succ_{i_6}) =\{s_2,s_4,s_3\}
\subseteq
L_{s_6}(\succ_{i_6}') =\{s_1,s_2,s_4,s_3\},
\]
and the relative order
\[
s_2\succ_{i_6}s_4\succ_{i_6}s_3
\]
of the schools originally below $s_6$ is preserved.}

\textcolor{black}{The DA outcome remains unchanged under this report, but JBC now generates
\[
\jbc(P')=(i_1-s_1,i_2-s_2,i_3-s_4,i_4-s_3,i_5-s_6,i_6-s_5).
\]
Hence $i_6$ improves from $s_6$ to $s_5$, her true second choice.
Therefore the upper manipulation is profitable, and JBC is not upper-manipulation-proof.}

\end{document}

%% file: prelim.tex
\usepackage{mathtools}
\usepackage{bm}
\usepackage{tikz}
\usepackage{forest}
\usetikzlibrary{arrows.meta}
\usepackage{booktabs} % For formal tables
\usepackage[ruled]{algorithm2e} % For algorithms

\SetAlFnt{\small}
\SetAlCapFnt{\small}
\SetAlCapNameFnt{\small}
\IncMargin{-\parindent}

\usepackage{amsthm}
\usepackage[yyyymmdd,24hr]{datetime}
\newcounter{exctr}
\renewcommand{\theexctr}{\arabic{exctr}}

\newtheorem{theorem}{Theorem}

\newtheorem{lemma}{Lemma}
\newtheorem{proposition}{Proposition}
\newtheorem{corollary}{Corollary}

\newtheorem{definition}{Definition}

\newtheorem{example}{Example}

\makeatletter
\renewcommand{\thmhead}[3]{%
	\thmheadfont{#1\ #2}%
	\thm@notefont{: [#3]}%
}
\makeatother

\newcommand{\UU}{\mathcal{U}}

\newcommand{\IR}{\mathcal{I}}
\newcommand{\rk}{\textup{rk}} 
\newcommand{\da}{\mathtt{DA}}

\newcommand{\eada}{\mathtt{EADA}} 
\newcommand{\jbc}{\mathtt{JBC}}

\newcommand{\BB}{\mathcal{B}}

\usepackage{microtype}
\usepackage{setspace}
\definecolor{darkblue}{rgb}{0.0, 0.0, 0.55}
\newcommand{\josue}[1]{\textcolor{black}{#1}}

%% file: aug_paths.tex
	\begin{tikzpicture}[scale=0.38, thick,
	every node/.style={draw, circle, minimum size=0.35cm, font=\scriptsize}]
	
	% ============ PANEL (a): ALL SELF-LOOPS ============
	
	\def\xa{0}
	
	% Left column
	\foreach \i/\y in {1/10, 2/8, 3/6, 4/4, 5/2, 6/0} {
		\node (A_L\i) at (\xa, \y) {$i_\i$};
	}
	% Right column
	\foreach \i/\y in {1/10, 2/8, 3/6, 4/4, 5/2, 6/0} {
		\node (A_R\i) at (\xa+6, \y) {$i_\i$};
	}
	
	% All self-loops (dotted gray)
	\draw[dotted, gray, thick] (A_L1) -- (A_R1);
	\draw[dotted, gray, thick] (A_L2) -- (A_R2);
	\draw[dotted, gray, thick] (A_L3) -- (A_R3);
	\draw[dotted, gray, thick] (A_L4) -- (A_R4);
	\draw[dotted, gray, thick] (A_L5) -- (A_R5);
	\draw[dotted, gray, thick] (A_L6) -- (A_R6);
	
	\draw[ gray, thick] (A_L1) -- (A_R2);		
	\draw[ gray, thick] (A_L1) -- (A_R4);		
	\draw[ gray, thick] (A_L3) -- (A_R6);		
	\draw[gray, thick] (A_L5) -- (A_R3);		
	\draw[gray, thick] (A_L5) -- (A_R1);		
	\draw[ gray, thick] (A_L4) -- (A_R5);		
	% Labels
	\node[draw=none, font=\small] at (\xa, 11.5) {\textbf{Out}};
	\node[draw=none, font=\small] at (\xa+6, 11.5) {\textbf{In}};
	\node[draw=none, font=\small] at (\xa+3, 13) {\textbf{(a) Initial}};
	
	% ============ PANEL (b): AFTER JBC + AUGMENTING PATH ============
	
	\def\xb{14}
	
	% Left column
	\foreach \i/\y in {1/10, 2/8, 3/6, 4/4, 5/2, 6/0} {
		\node (B_L\i) at (\xb, \y) {$i_\i$};
	}
	% Right column
	\foreach \i/\y in {1/10, 2/8, 3/6, 4/4, 5/2, 6/0} {
		\node (B_R\i) at (\xb+6, \y) {$i_\i$};
	}
	
	% JBC matching (blue, thick)
	\draw[blue, ultra thick] (B_L1) -- (B_R4);
	\draw[blue, ultra thick] (B_L4) -- (B_R5);
	\draw[blue, ultra thick] (B_L5) -- (B_R1);
	
	% Self-loops (dotted gray)
	\draw[dotted, gray, thick] (B_L2) -- (B_R2);
	\draw[dotted, gray, thick] (B_L3) -- (B_R3);
	\draw[dotted, gray, thick] (B_L6) -- (B_R6);
	
	% Augmenting path (dashed orange)
	\draw[->, orange, dashed, ultra thick] (B_L2) -- (B_R1);
	\draw[->, orange, dashed, ultra thick] (B_L5) -- (B_R3);
	\draw[->, orange, dashed, ultra thick] (B_L3) -- (B_R6);
	\draw[->, orange, dashed, ultra thick] (B_L6) -- (B_R4);
	\draw[->, orange, dashed, ultra thick] (B_L1) -- (B_R2);
	
	% Labels
	\node[draw=none, font=\small] at (\xb, 11.5) {\textbf{Out}};
	\node[draw=none, font=\small] at (\xb+6, 11.5) {\textbf{In}};
	\node[draw=none, font=\small] at (\xb+3, 13) {\textbf{(b) JBC + augmenting path}};
	
	% ============ PANEL (c): OPTIMAL MATCHING ============
	
	\def\xc{28}
	
	% Left column
	\foreach \i/\y in {1/10, 2/8, 3/6, 4/4, 5/2, 6/0} {
		\node (C_L\i) at (\xc, \y) {$i_\i$};
	}
	% Right column
	\foreach \i/\y in {1/10, 2/8, 3/6, 4/4, 5/2, 6/0} {
		\node (C_R\i) at (\xc+6, \y) {$i_\i$};
	}
	
	% Optimal matching (red, thick)
	\draw[red, ultra thick] (C_L1) -- (C_R2);
	\draw[red, ultra thick] (C_L2) -- (C_R1);
	\draw[red, ultra thick] (C_L3) -- (C_R6);
	\draw[red, ultra thick] (C_L4) -- (C_R5);
	\draw[red, ultra thick] (C_L5) -- (C_R3);
	\draw[red, ultra thick] (C_L6) -- (C_R4);
	
	% Labels
	\node[draw=none, font=\small] at (\xc, 11.5) {\textbf{Out}};
	\node[draw=none, font=\small] at (\xc+6, 11.5) {\textbf{In}};
	\node[draw=none, font=\small] at (\xc+3, 13) {\textbf{(c) SJBC+}};
	
\end{tikzpicture}

%% file: twomatchings.tex
\begin{minipage}{0.45\textwidth}
	\centering
	\begin{tikzpicture}[scale=0.7, ->, node distance=2.5cm, every node/.style={draw, circle, minimum size=1cm}, thick]
		\def\radius{3}
		
		% Define nodes in a circular layout (only 6 nodes)
		\foreach \i/\name in {1/i_1, 2/i_2, 3/i_3, 4/i_4, 5/i_5, 6/i_6} {
			\node[circle, draw] (\name) at ({360/6 * (\i-1)}:\radius) {$\name$}; 
		}
		
		% Color nodes 4,5,1,6 with blue borders
		\foreach \i in {1,4,5,6} {
			\node[circle, draw=blue, minimum size=1cm] (i_\i) at ({360/6 * (\i-1)}:\radius) {$i_\i$}; 
		}
		
		% Draw edges based on the given connections
		% 1->6,4,2,3,5
		\draw[ultra thick, blue] (i_1) to (i_6);
		\draw[thick] (i_1) to (i_4);
		\draw[thick] (i_1) to (i_2);
		\draw[thick] (i_1) to (i_3);
		\draw[thick] (i_1) to (i_5);
		
		% 2->1
		\draw[thick] (i_2) to (i_1);
		
		% 3->6
		\draw[thick] (i_3) to (i_6);
		
		% 4->5
		\draw[ultra thick, blue, bend right] (i_4) to (i_5);
		
		% 5->1,4,6,3
		\draw[ultra thick, blue, bend right=10] (i_5) to (i_1);
		\draw[thick] (i_5) to (i_4);
		\draw[thick] (i_5) to (i_6);
		\draw[thick] (i_5) to (i_3);
		
		% 6->4
		\draw[ultra thick, blue] (i_6) to (i_4);
	\end{tikzpicture}
\end{minipage}
\hfill
\begin{minipage}{0.45\textwidth}
	\centering
	\begin{tikzpicture}[scale=0.7, ->, node distance=2.5cm, every node/.style={draw, circle, minimum size=1cm}, thick]
		\def\radius{3}
		
		% Define nodes in a circular layout (only 6 nodes)
		\foreach \i/\name in {1/i_1, 2/i_2, 3/i_3, 4/i_4, 5/i_5, 6/i_6} {
			\node[circle, draw] (\name) at ({360/6 * (\i-1)}:\radius) {$\name$}; 
		}
		
		% Color nodes 1-6 with red borders
		\foreach \i in {1,2,3,4,5,6} {
			\node[circle, draw=red, minimum size=1cm] (i_\i) at ({360/6 * (\i-1)}:\radius) {$i_\i$}; 
		}
		
		% Draw edges based on the given connections
		% 1->6,4,2,3,5
		\draw[thick] (i_1) to (i_6);
		\draw[thick] (i_1) to (i_4);
		\draw[ultra thick, red] (i_1) to (i_2);
		\draw[thick] (i_1) to (i_3);
		\draw[thick] (i_1) to (i_5);
		
		% 2->1
		\draw[ultra thick, red] (i_2) to (i_1);
		
		% 3->6
		\draw[ultra thick, red] (i_3) to (i_6);
		
		% 4->5
		\draw[ultra thick, red, bend right] (i_4) to (i_5);
		
		% 5->1,4,6,3
		\draw[thick, bend right=10] (i_5) to (i_1);
		\draw[thick] (i_5) to (i_4);
		\draw[thick] (i_5) to (i_6);
		\draw[ultra thick, red] (i_5) to (i_3);
		
		% 6->4
		\draw[ultra thick, red] (i_6) to (i_4);
	\end{tikzpicture}
\end{minipage}

%% file: tree.tex
\resizebox{\linewidth}{!}{%
	\begin{forest}
		for tree={
			grow'=south,
			parent anchor=south,
			child anchor=north,
			align=center,
			edge={-Latex},
			l sep=10mm,
			s sep=16mm,
			inner sep=2pt,
			font=\small,
		}
		[{$i_7,s_4$}
		% ===== YES (left) =====
		[{$i_3,s_6$}, edge label={node[midway,left,font=\scriptsize]{yes}}
		% YES (left)
		[{$i_5,s_6$}, edge label={node[midway,left,font=\scriptsize]{yes}}
		[{$1\to 6\to 4\to 5\to 1$}, draw, rectangle, rounded corners, 
		edge label={node[midway,left,font=\scriptsize]{yes}}]
		[{$1\to 4\to 5\to 1$}, draw, rectangle, rounded corners, 
		edge label={node[midway,right,font=\scriptsize]{no}}]
		]
		% NO (right)
		[{$1\to 4\to 5\to 1$}, draw, rectangle, rounded corners, 
		edge label={node[midway,right,font=\scriptsize]{no}}]
		]
		% ===== NO (right) =====
		[{$i_4,s_5$}, edge label={node[midway,right,font=\scriptsize]{no}}
		% YES (left)
		[{$i_3,s_6$}, edge label={node[midway,left,font=\scriptsize]{yes}}
		[{$1\to 5\to 1$}, draw, rectangle, rounded corners,
		edge={-Latex, dashed},
		edge label={node[midway,left,font=\scriptsize]{$\cdots$}}
		]
		]
		% NO (right)
		[{$i_5,s_1$}, edge label={node[midway,right,font=\scriptsize]{no}}
		% YES (left)
		[{$i_3,s_6$}, edge label={node[midway,left,font=\scriptsize]{yes}}
		[{$1\to 2\to 1$}, draw, rectangle, rounded corners,
		edge={-Latex, dashed},
		edge label={node[midway,left,font=\scriptsize]{$\cdots$}}
		]
		]
		% NO (right)
		[{$i_3,s_6$}, edge label={node[midway,right,font=\scriptsize]{no}}
		[{no\\cycles}, draw, circle, minimum size=10mm,
		edge={-Latex, dashed},
		edge label={node[midway,right,font=\scriptsize]{$\cdots$}}
		]
		]
		]
		]
		]
	\end{forest}%
}